\documentclass[sigplan,10pt]{acmart}
\usepackage{times, datetime}
\usepackage{fullpage}
\usepackage{url, hyperref}
\usepackage{amsmath, amsthm, amsfonts, amssymb}
\usepackage{graphicx}
\usepackage{fancyvrb}
\usepackage{listings}
\usepackage{xcolor}
\usepackage{subcaption}
\usepackage[linesnumbered,ruled,vlined]{algorithm2e}
\usepackage[T1]{fontenc}
\usepackage[scaled]{beramono}
\usepackage{enumitem}
\usepackage{todonotes}
\usepackage{tabularx, multirow}
\usepackage[normalem]{ulem}

\renewcommand\footnotetextcopyrightpermission[1]{}
\renewcommand{\textrightarrow}{$\rightarrow$}

\setlist[enumerate]{noitemsep, topsep=0pt}

\newcommand{\name}{Tofu\xspace}
\newcommand{\code}[1]{\texttt{\small #1}}
\newcommand{\kw}[1]{\code{\color{blue}#1}}
\newcommand{\pnr}{partition-n-reduce\xspace}
\newcommand{\Pnr}{Partition-n-reduce\xspace}
\newcommand{\wnet}{WResNet\xspace}
\newcommand{\secref}[1]{Sec~\ref{#1}}

\newtheorem{theorem}{Theorem}
\newtheorem{lemma}{Lemma}
\newtheorem{corollary}{Corollary}
\newtheorem*{thm:greedy}{Theorem \ref{thm:greedy}}

\def\nocolormarks{}

\ifx\nocolormarks\undefined
\newcommand{\textred}[1]{{\color{red}{#1}}}
\newcommand{\pgwrapper}[3]{\begingroup {\color{#1} #2: #3} \endgroup}
\else
\newcommand{\textred}[1]{#1}
\newcommand{\pgwrapper}[3]{}
\fi

\def\hn{\sffamily\selectfont}
\newcommand{\mpfont}{\hn\scriptsize}
\ifx\nocolormarks\undefined
\newcommand{\MPworker}[2]{\unskip{\color{#1}\vrule\vrule}{\marginpar{\raggedright\color{#1}\mpfont #2}}}
\else
\newcommand{\MPworker}[2]{\unskip}
\fi

\newcommand{\REV}[1]{\MPworker{blue}{#1}}

\ifx\nocolormarks\undefined
\newcommand{\changebars}[2]{%
	({\color{magenta}\sout{#1}}){\color{red}\em\begingroup{#2}\endgroup}}
\else
\newcommand{\changebars}[2]{#1}
\fi

\begin{document}

\title{Supporting Very Large Models using Automatic \\Dataflow Graph Partitioning}

\author{Minjie Wang}
\affiliation{New York University}

\author{Chien-chin Huang}
\affiliation{New York University}

\author{Jinyang Li}
\affiliation{New York University}
\date{}

\begin{abstract}
	
	This paper presents \name, a system that partitions very large DNN models across multiple GPU devices
	to reduce per-GPU memory footprint.  \name is designed to partition a dataflow
	graph of fine-grained tensor operators \changebars{used by platforms like MXNet and TensorFlow}{in order to work transparently with a
		general-purpose deep learning platform like MXNet}.  In order to
	automatically partition each operator, we propose to describe the
	semantics of an operator in a simple language 
	\changebars{inspired by Halide}{which represents tensors as lambda
		functions mapping from tensor coordinates to values}.  To optimally
	partition different operators in a dataflow graph, \name uses a
	recursive search algorithm that minimizes the total communication cost.
	Our experiments on an 8-GPU machine show that \name enables the
	training of very large CNN and RNN models.  It also achieves 25\% -
	400\% speedup over alternative approaches to train very large models.
	
\end{abstract}

\copyrightyear{2019}
\acmYear{2019}
\setcopyright{acmlicensed}
\acmConference[EuroSys '19]{Fourteenth EuroSys Conference 2019}{March
	25--28, 2019}{Dresden, Germany}
\acmBooktitle{Fourteenth EuroSys Conference 2019 (EuroSys '19), March 25--28,
	2019, Dresden, Germany}
\acmPrice{15.00}
\acmDOI{10.1145/3302424.3303953}
\acmISBN{978-1-4503-6281-8/19/03}

\ccsdesc[300]{Computer systems organization~Neural networks}
\ccsdesc[300]{Computer systems organization~Data flow architectures}


\maketitle

\section{Introduction}\label{s:intro}

The deep learning community has been using larger deep
neural network (DNN) models to achieve higher accuracy on more complex tasks over the
past few years~\cite{wideresnet,google:nmt}. Empirical evidence shows that, since the 80s,
the number of parameters in the state-of-the-art neural network has doubled
roughly every 2.4 years~\cite{deeplearningbook}, enabled by hardware
improvements and the availability of large datasets.  As deployed DNN models
remain many orders of magnitude smaller than that of a mammalian brain,
there remains much room for growth. However, the size of a DNN model that
can be explored today is constrained by the limited GPU device memory.

There have been many efforts to tackle the problem of limited GPU device
memory.  Some proposals try to fit larger models into a single GPU, e.g. by
using the much larger CPU memory as a swap area for the GPU~\cite{meng2017training}
or by discarding intermediate results to save memory at the cost of
re-computation~\cite{gruslys2016memory,martens2012training,chen2016training}.
Another promising solution is to partition a DNN model across multiple GPU
devices. Doing so reduces per-GPU memory footprint and comes with the
additional benefit of parallel speedup.  This is commonly referred to as
``model parallelism'' in the literature.

A DNN model consists of a large number of layers, each parameterized by its own
weights.  There are two approaches to realize model parallelism.  One approach
is to assign the computation of different layers to different devices.  
The second approach is to partition the tensors to parallelize 
each layer across devices.  For very large DNN models, tensor partitioning is
the better approach; not only it results in balanced per-GPU memory usage but
also it necessary for speeding up
popular models such as CNNs.



Tensor partitioning has been explored by existing work 
as a means for achieving parallel speedup~\cite{dean:nn,coates2013deep,projectadam}
or saving memory access energy~\cite{yang2016systematic,neurocube}.
Recent proposals~\cite{tetris, jia2018exploring, jia2018beyond} support
partitioning a tensor along multiple dimensions and can automatically search
for the best partition dimensions. The major limitation is that these proposals
partition at the coarse granularity of individual DNN layers, \textred{such as fully-connected 
and 2D convolution layers}. As such, they
either develop specialized implementation for specific models~\cite{tetris,
coates2013deep} or allow only a composition of common DNN
layers~\cite{jia2018exploring, jia2018beyond, dean:nn, projectadam}.

However, the vast majority of DNN development and deployment today occur on
general-purpose deep learning platforms such as TensorFlow~\cite{tensorflow},
MXNet~\cite{mxnet}, PyTorch~\cite{pytorch}. These platforms represent
computation as a dataflow graph of fine-grained tensor operators, \textred{such as matrix 
multiplication, various types of convolution and element-wise operations etc}.
Can we support tensor partitioning on one of these general-purpose platforms?
To do so, we have built the \name system \textred{to automatically partition
the input/output tensors of each operator in the MXNet dataflow system.}
This approach, which we call {\em operator partitioning}, is more fine-grained
than layer partitioning. \textred{While we have built \name's prototype to work
with MXNet, \name's solution is general and could potentially be applied to
other dataflow systems such as TensorFlow.}

In order to partition a dataflow graph of operators, \name
must address two challenges. 1) How to partition the input/output tensors and parallelize the execution
an individual operator? What are the viable partition dimensions? 2) how to optimize the partitioning of different operators 
for the overall graph? Both
challenges are made difficult by the fine-grained approach of partitioning
operators instead of layers. For the first challenge, existing
work~\cite{tetris,jia2018exploring,jia2018beyond} manually discover how to partition a few common
layers.  However, a dataflow framework supports a large and growing collection of operators 
(e.g. 139 in MXNet), intensifying the manual efforts.
Manual discovery is also error-prone, and can miss certain partition strategies.
\textred{For example, \cite{jia2018exploring} misses a crucial partition strategy 
that can significantly reduce per-worker memory footprint (\secref{ss:eval-partition}).}
For the second challenge, existing proposals use greedy or dynamic-programming based
algorithms~\cite{tetris,jia2018exploring} or stochastic searches~\cite{jia2018beyond}.
As the graph of operators is more complex and an order of magnitude larger than the graph of
layers (e.g. the graph for training a 152-layer ResNet has >1500 operators in MXNet),
these algorithms become inapplicable or run too slowly (\secref{s:dataflow}, Table~\ref{tbl:improve}).


%
\name introduces novel solutions to address the above mentioned challenges. 
To enable the automatic discovery of an operator's partition dimensions, \name requires
developers to specify what the operator computes using a lightweight description
language called TDL.  Inspired by Halide~\cite{halide}, TDL describes tensor
computation by specifying the output tensor value at each index with simple
expressions on the input tensors. \textred{The Halide-style description is useful because 
it makes explicit which input tensor regions are needed in order to compute a
specific output tensor region.} Thus, 
\name can statically analyze an operator's TDL
description using symbolic execution to determine \textred{what input regions must be transferred among GPUs when tensors are divided along a specific partition dimension.} 
To partition each tensor in the overall dataflow
graph, we propose several techniques to shrink the search space. These include
a recursive search algorithm which partitions the graph among only two workers at
each recursive step, and graph coarsening by grouping related operators.


We have implemented a prototype of \name in MXNet and evaluated its
performance on a single machine with eight GPUs. Our experiments use
large DNN models including Wide ResNet~\cite{wideresnet} and Multi-layer
Recurrent Neural Networks~\cite{rnnlimit}, most of which do not fit in a single GPU's
memory. 
Compared with other approaches to train large models, \name's training
throughput is 25\% - 400\% higher. 

To the best of our knowledge, \name is the first system to automatically
partition a dataflow graph of fine-grained tensor operators.  Though promising,
\name has several limitations (\secref{s:discussion}).  Some operators (e.g.
Cholesky) cannot be expressed in TDL and thus cannot be automatically
partitioned.  The automatically discovered partition strategies do not exploit
the underlying communication topology.  \name is also designed for very large
DNN models.
For moderately sized models that do fit in the memory of a single
GPU,  \textred{\name's approach of operator partitioning are likely no better than 
the much simpler approach of data parallelism.}
Removing these limitations requires further research.

\section{Background}
\label{sec:background}

{\bf The problem.} Training very large DNN models is limited
by the size of GPU device memory today. Compared with CPU memory,
GPU memory has much higher bandwidth but also smaller capacity, ranging from
12GB (NVIDIA K80) to 16GB (NVIDIA Tesla V100).  Google's TPU hardware has
similar limitations, with 8GB attached to each TPU
core~\cite{tpu:architecture}.

Partitioning each tensor in the DNN computation across multiple 
devices can lower per-GPU memory footprint, thereby allowing very large models
to be trained. When partitioning 
across $k$ devices, each device roughly consumes $\frac{1}{k}$ times the total
memory required to run the computation on one device.  Furthermore,
partitioning also has the important benefit of performance speedup via parallel
execution. As most DNN development today is done 
on dataflow platforms such as TensorFlow and MXNet, our goal is to
automatically partition the tensors and parallelize the operators in a dataflow
graph to enable the training of very large DNN models. The partitioning should
be completely transparent to the user: the same program written for a single
device can also be run across devices without changes.


\vspace{0.1in}
{\bf System setting.} When tensors are partitioned,
workers must communicate with each other to fetch the data needed for
computation. The amount of bytes transferred divided by the computation time
forms a lower bound of the communication bandwidth required to achieve
competitive performance. For training very large DNNs on fast GPUs, the aggregate bandwidth
required far exceeds the network bandwidth in deployed GPU clusters (e.g.
Amazon's EC2 GPU instances have only 25Gbps aggregate bandwidth).  Thus, for
our implementation and evaluation, we target a single machine with multiple GPU
devices.

\section{Challenges and our approach}
\label{s:challenges}

\begin{figure}
\centering
\lstset{%
	basicstyle=\scriptsize\ttfamily,%
	emphstyle={\color{blue}\bfseries},
	commentstyle={\color{teal}\itshape},
	keywords={lambda,def,return,for,in},
}%
\begin{lstlisting}
def conv1d(data, filters):
 for b in range(output.shape[0]): #b is batch dimension
   for co in range(output.shape[1]): #co is output channel
    for x in range(output.shape[2]): #x is output pixel
     for ci in range(filters.shape[0]): #di is input channel
      for dx in range(filters.shape[2]): #dx is filter window
       output[b, co, x] += data[b, ci, x+dx] 
                           * filters[ci, co, dx]
\end{lstlisting}
\caption{The naive implementation of \code{conv1d} in Python.}\label{fig:conv1d}
\end{figure}



In order to partition a dataflow graph of operators, we must tackle the two challenges
mentioned in~\secref{s:intro}.
We discuss
these two challenges in details and explain at a high level how \name solves
them.

\subsection{How to partition a single operator?}
To make the problem of automatic partitioning tractable, we consider only a restricted
parallelization pattern, which we call ``{\bf \pnr}''.  
Suppose operator $c$ computes output tensor $O$.  Under \pnr, $c$ can be
parallelized across two workers by executing the {\em same} operator on each
worker using smaller inputs.  The final output tensor $O$ can be obtained
from the output tensors of both workers ($O_1$, and $O_2$) in one of the two ways.
1) $O$ is the concatenation of $O_1$ and $O_2$ along some dimension.  2) $O$
is the element-wise reduction of $O_1$ and $O_2$.  \Pnr is crucial for
automatic parallelization because it allows an operator's existing single-GPU 
implementation to be re-used for parallel execution.  Such implementation often
belongs to a highly optimized closed-source library (e.g.  cuBLAS, cuDNN).  

\Pnr is not universally applicable, e.g. Cholesky~\cite{madlinq} cannot be
parallelized this way.  Nor is \pnr optimal.  One can achieve more efficient
communication with specialized parallel algorithms (e.g. Cannon's
algorithm~\cite{cannon} for matrix multiplication) than with \pnr.
Nevertheless, the vast majority of operators can be parallelized using \pnr
(\secref{ss:tdl}) and have good performance.

Tensors used in DNNs have many dimensions so there are
potentially many different ways to parallelize an operator.
Figure~\ref{fig:conv1d} shows an example
operator, \code{conv1d}, which computes 1-D convolution over \code{data} using
\code{filters}.   The 3-D \code{data} tensor
contains a batch (\code{b}) of 1-D pixels with \code{ci} input channels.
The 3-D \code{filters} tensor contains a convolution
window for each pair of \code{ci} input and \code{co} output channel.
The 3-D \code{output} tensor contains the convolved
pixels for the batch of data on all output channels.  

There are many ways to parallelize \code{conv1d} using \pnr;
Figure~\ref{fig:parconv1d} shows two of them.  In
Figure~\ref{fig:parconv1d}(a), the final output is a concatenation (along the
\code{b} dimension) of output tensors computed by each worker.  Each worker
reads the entire \code{filters} tensor and half of the \code{data} tensor. In
Figure~\ref{fig:parconv1d}(b), the final output is a reduction (sum) of each
worker's output.  Figure~\ref{fig:conv1d} shows what input tensor region each
work reads from. If tensors are partitioned, workers must perform remote data
fetch.


\begin{figure}
\centering
\includegraphics[scale=0.5]{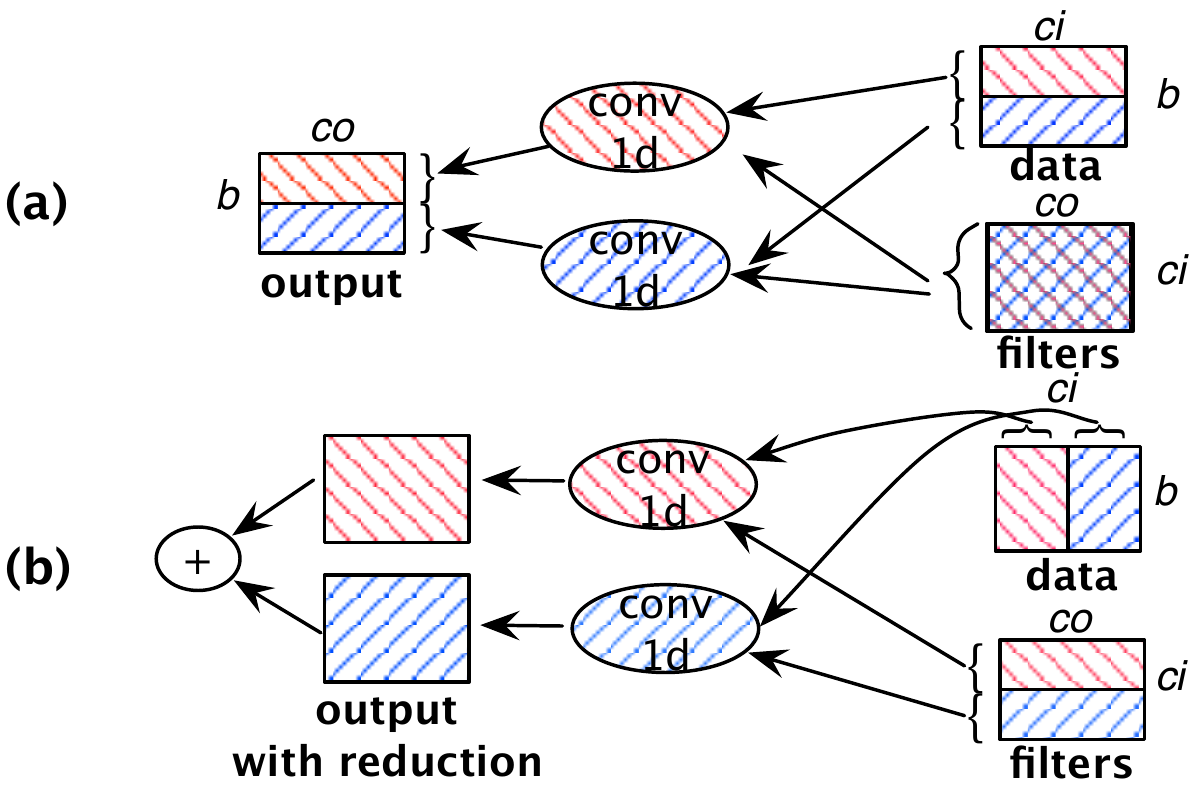}
\caption{Two of several ways to parallelize \code{conv1d} according to \pnr.
Each 3D tensor is represented as a 2D matrix of vectors.
Different stripe patterns show the input tensor regions required by different workers.}
\label{fig:parconv1d}
\end{figure}

Prior work~\cite{tetris, jia2018exploring, jia2018beyond} manually 
discovers the partition strategies for a few common DNN layers.
Some~\cite{jia2018exploring, jia2018beyond} have ignored the strategy that uses
output reduction (i.e.  Figure~\ref{fig:parconv1d}(b)), which we show to have
performance benefits later (\secref{ss:eval-partition}).  Manual discovery is 
tedious for a dataflow system with a large number of operators (341 and 139 in
TensorFlow and MXNet respectively). Can one support automatic discovery instead? 

\vspace{0.1in}
{\bf Our approach.}  \name analyzes the access pattern of an operator to 
determine all viable partition strategies. As such, we require the
developer of operators to provide a succinct description of what each operator
computes in a light-weight language called TDL (short for Tensor Description
Language).  An operator's TDL description is separate from its implementation.  The
description specifies at a high-level how the output tensor is derived from its
inputs, without any concern for algorithmic or architectural optimization,
which are handled by the operator's implementation.  We can statically analyze
an operator's TDL description to determine how to partition it along different
dimensions. \secref{s:descr} describes this part of \name's design in
details.

\subsection{How to optimize partitioning for a graph?}
As each operator has several partition strategies, there are combinatorially
many choices to partition each tensor in the dataflow graph, each of which has
different execution time and per-GPU memory consumption.

It is a NP-hard problem to partition a general dataflow graph for optimal
performance~\cite{kennedy1998automatic, kremer1993np, li1990index, li1991data}.
Existing proposals use greedy or dynamic-programming algorithm to optimize a mostly linear
graph of layers~\cite{tetris, jia2018exploring},
or perform stochastic searches~\cite{jia2018beyond, mirhoseini2017device, placement:iclr18}
for general graphs.  The former approach is faster, but still impractical when applied 
on fine-grained dataflow graphs. In particular, its running time is proportional
to the number of ways an operator can be partitioned.  When there are $2^m$
GPUs, each input/output tensor of an operator can be partitioned along a
combination of any 1, 2, ..., or $m$ dimensions, thereby dramatically
increasing the number of partition strategies and exploding the search time.



\vspace{0.1in}
{\bf Our approach.} We use an existing dynamic programming (DP) algorithm~\cite{jia2018exploring} 
in our search and propose several key techniques to make it
practical.  First, we leverage the unique characteristics of DNN computation to
``coarsen'' the dataflow graph and shrink the search space.  These include
grouping the forward and backward operations, and coalescing element-wise or
unrolled operators.  Second, to avoid blowing up the search space in the face
of many GPUs, we apply the basic search algorithm recursively.  In
each recursive step, the DP algorithm only needs to partition
each tensor in the coarsened graph among two ``groups'' (of GPUs).  
\secref{s:dataflow} describes this part of \name's design in details.

\section{Partitioning a single operator}
\label{s:descr}

\begin{figure}[t]
\centering
\lstset{%
	basicstyle=\scriptsize\ttfamily,%
	emphstyle={\color{blue}\bfseries},
	keywords={lambda,def,return},
}%
\begin{lstlisting}
@tofu.op
def conv1d(data, filters):
  return lambda b, co, x:
    Sum(lambda ci, dx: data[b, ci, x+dx]*filters[ci, co, dx])

@tofu.op
def batch_cholesky(batch_mat):
  Cholesky = tofu.Opaque()
  return lambda b, i, j: Cholesky(batch_mat[b, :, :])[i,j]
\end{lstlisting}
\caption{Example TDL descriptions.}\label{fig:conv1dtdl}
\end{figure}

This section describes TDL (\secref{ss:tdl}) and its analysis (\secref{ss:region}).  

\subsection{Describing an operator}\label{ss:tdl}
Our Tensor Description Language (TDL) is inspired by
Halide\cite{halide}.  The core idea is ``\emph{tensor-as-a-lambda}'', i.e.
we represent tensors as lambda functions that map from coordinates (aka index
variables) to values, expressed as a TDL expression.  TDL expressions are
side-effect free and include the following:

\begin{itemize}[leftmargin=*]
\setlength\itemsep{0em}
\item Index variables (i.e. arguments of the lambda function).
\item Tensor elements (e.g. \code{filters[ci, co, dx]}).
\item Arithmetic operations involving constants, index variables, tensor elements or TDL expressions.
\item Reduction over a tensor along one or more dimensions.
\end{itemize}

%
%
\emph{Reducers} are commutative and associative functions that aggregate 
elements of a tensor along one or more dimensions.  \name supports 
\code{Sum}, \code{Max}, \code{Min} and \code{Prod} as built-in reducers.  It
is possible to let programmers define custom reducers,  but 
we have not encountered the need to do so.



We implemented TDL as a DSL using Python.  As an example,
Figure~\ref{fig:conv1dtdl} shows the description of \code{conv1d}, whose output
is a 3D tensor defined by \code{\kw{lambda} b, co, x: ...} Each element of the
output tensor is the result of reduction (\code{Sum}) over an internal 2D
tensor (\code{\kw{lambda} ci, dx: ...}) over both \code{ci} and \code{dx}
dimensions. 

{\bf Opaque function.} \textred{We have deliberately designed TDL to be simple and not
Turing-complete. For example, 
TDL does not support loops or recursion, and thus cannot express sophisticated 
computation such as Cholesky decomposition. In such cases, we represent
the computation as an \emph{opaque function}.}  Sometimes, such an operator has
a batched-version that can be partitioned along the batch
dimension. Figure~\ref{fig:conv1dtdl} shows the TDL description of the
operator \code{batch\_cholesky}. The output is a 3-D tensor (\code{\kw{lambda}
b,i,j:...}) where the element at $(b,i,j)$ is defined to be the $(i,j)$ element
of the matrix obtained from performing Cholesky on the $b$-th slice of the
input tensor. Note that, \code{batch\_mat[b, :, :]} represents the
\code{b}$^{th}$ slice of the \code{batch\_mat} tensor. It is syntactic
sugar for the lambda expression \code{\kw{lambda} r, c: batch\_mat[b, r, c]}.

{\bf Describing MXNet operators in TDL.} 
Ideally, operator developers should write TDL descriptions.  As \name is
meant to work with an existing dataflow system (MXNet), we have written
the descriptions ourselves as a way to bootstrap.
We found that TDL can describe 134 out of 139 MXNet v0.11 operators. Out of these,
77 are simple element-wise operators; 2 use the opaque function primitive, and
\textred{11 have output reductions}.
It takes one of the authors one day to write all these descriptions; most of them
have fewer than three LoC.
\textred{Although we did not build \name's prototype for TensorFlow, we did
investigate how well TDL can express TensorFlow operators.  We found that TDL
can describe 257 out of 341 TensorFlow operators. Out of these, 140 are
element-wise operators; 22 use the opaque function.  For those operators that
cannot be described by TDL, they belong to three categories: sparse tensor
manipulations, operators with dynamic output shapes and operators requiring
data-dependent indexing. MXNet has no operators in the latter two categories.}


{\bf TDL vs. other Halide-inspired language.} Concurrent with our work, 
TVM~\cite{tvm:osdi18} and TC~\cite{tc:fb} are two other Halide-inspired DSLs.
Compared to these DSLs, TDL is designed for a different purpose. Specifically, we
use TDL to analyze an operator's partition strategies while TVM and TC
are designed for code generation to different hardware platforms. The different usage scenarios lead to two design differences.
First, TDL does not require users to write intricate execution schedules -- code for
describing how to perform loop transformation, caching, and mapping to
hardwares, etc.  
Second, TDL supports opaque functions that let users elide certain details
of the computation that are not crucial for analyzing how the operator can be
partitioned.  


\subsection{Analyzing TDL Descriptions}
\label{ss:region}

\name analyzes the TDL description of an operator to discover its basic
partition strategies.  A basic partition strategy parallelizes an operator for
2 workers only.  Our search algorithm uses basic strategies recursively to
optimize partitioning for more than two workers (\secref{ss:dataflow-recursive}).

A partition strategy can be specified by describing the input tensor
regions required by each worker to perform its ``share'' of the
computation.  This information is used later by our search algorithm to
optimize partitioning for the dataflow graph and to generate the partitioned
graph in which required data is fetched from different workers.

Obtaining input regions from a TDL description is straightforward
if tensor shapes are known.
For example, consider the following simple description:
\begin{lstlisting}
def shift_two(A): B = lambda i : A[i+2]; return B
\end{lstlisting}
Suppose we want to partition along output dimension \code{i}. Given \code{i}'s
concrete range, say $[0,9]$,  we can compute that the worker needs \code{A}'s
data over range $[2,6]$ (or $[7,11]$) in order to compute \code{B} over range
$[0,4]$ (or $[5,9]$).  

Analyzing with concrete ranges is hugely
inefficient as a dataflow graph can contain thousands of operators, many of
which are identical except for their tensor shapes (aka index ranges).
Therefore, we perform TDL analysis in the 
abstract domain using \emph{symbolic interval analysis}, a technique
previously used for program variable analysis\cite{venet2012gauge},
boundary checking\cite{rugina2000symbolic}, parameter validation\cite{wu2014abstract}.


{\bf Symbolic interval analysis.} Suppose the 
output tensor of an operator has $n$ dimensions and is of the form \code{lambda x1, ..., xn : ...}.
We consider the range of index variable \code{xi} to be [0,
$\mathcal{X}_i$], where $\mathcal{X}_i$ is a symbolic upper bound.  We then
symbolically execute the lambda function to calculate the symbolic intervals
indicating the range of access on the operator's input tensors. 

Symbolic execution should keep the range as precise as possible.  To
do so, we represent symbolic interval ($\mathcal{I}$) as an affine
transformation of all symbolic upper bounds,
\begin{equation} \label{eq:i}
	\mathcal{I}\triangleq[\Sigma_i l_{i}\mathcal{X}_i+c,~\Sigma_i u_{i}\mathcal{X}_i+c], l_{i},u_{i},c\in\mathbb{R}
\end{equation}
In equation~\ref{eq:i}, $l_{i}$, $u_{i}$ and $c$ are some constants.  Thus, we can represent
$\mathcal{I}$ as a vector of $2*n+1$ real values $\langle
l_{1},~...,~l_{n},u_{1},...,u_{n},~c\rangle$.  Let $\mathtt{ZV}[u_i=a]$ denote
a vector of all 0s except for the position corresponding to $u_i$ which has value $a$.  By default, lambda variable
\code{xi} for dimension $i$ is initialized to $\mathtt{ZV}[u_i=1]$.

\begin{figure}
\[\arraycolsep=1.4pt
\begin{array}{rcl}
\multicolumn{3}{c}{\text{TDL description: \code{{\color{blue} lambda} x1, ..., xi, ..., xn: ...}}} \\
\mathcal{I} &\triangleq& \langle l_{1},~...,~l_{n},~u_{1},~...,~u_{n},~c \rangle \\
\mathcal{I} \pm k, k\in\mathbb{R} &=& \langle l_{1},~...,~l_{n},~u_{1},~...,~u_{n},~c \pm k \rangle \\
\mathcal{I} \times k, k\in\mathbb{R} &=& \langle l_{1}k,~...,~l_{n}k,~u_{1}k,~...,~u_{n}k,~c*k \rangle \\
\mathcal{I} / k, k\in\mathbb{R} &=& \langle l_{1}/k,~...,~l_{n}/k,~u_{1}/k,~...,~u_{n}/k,~c/k \rangle \\
\mathcal{I} \pm \mathcal{I}' &=& \langle l_{1}\pm l'_{1},~...,~u_{1}\pm u'_{1},~...,~c\pm c' \rangle \\
\end{array}
\]
\caption{\name's symbolic interval arithmetic.}\label{fig:int-arith}
\end{figure}

Our representation can support affine transformation on
the intervals, as shown by the allowed interval arithmetic in
Figure~\ref{fig:int-arith}. Product or comparison between two intervals are
not supported and will raise an error. We did not encounter any 
such non-affine operations among MXNet operators.

{\bf Discover operator partition strategies.}
Using the symbolic interval analysis, we infer the input regions required
by each of the 2 workers for every partitionable dimension.  There are two cases.

{\em Case-1} corresponds to doing \pnr without the reduction step. In this case, 
each partition strategy corresponds to some output dimension.
Suppose we are to partition \code{conv1d}'s output tensor along dimension \code{b}.
We use two different initial intervals for lambda variable \code{b},
$\mathtt{ZV}[u_{\mathtt{b}}=\frac{1}{2}]$ and $\mathtt{ZV}[l_{\mathtt{b}}=\frac{1}{2}, u_{\mathtt{b}}={1}]$, in two separate analysis runs.
Each run calculates the input regions needed to compute half of the
output tensor.  The result shows that 
that each worker reads half of the
\code{data} tensor partitioned on the \code{b} dimension and all
of the \code{filter} tensor, as illustrated in Figure~\ref{fig:parconv1d}(a).
Similarly, the analysis shows how to partition the other output dimensions, \code{co} and \code{x}.
Partitioning along dimension \code{x} is commonly referred to as parallel convolution with ``halo exchange''~\cite{coates2013deep,tetris,neurocube}.


{\em Case-2} corresponds to doing \pnr with the reduction step. 
In this case, we partition along a reduction dimension. In the example of Figure~\ref{fig:conv1dtdl}, 
the reduction dimensions corresponding to \code{ci} and \code{dx} in 
\code{Sum(\kw{lambda} ci, dx: ...)}.  The analysis will determine that, when
partitioning along \code{ci}, each partially reduced tensor
will require half of the \code{data} tensor partitioned on the
second dimension and half of the \code{filter} tensor
partitioned on the first dimension, as shown in
Figure~\ref{fig:parconv1d}(b).  Similar analysis is also done
for dimension \code{dx}. Out of 47 non-element-wise MXNet operators
describable by TDL, 11 have at least one reduction
dimension.

\section{Partitioning the dataflow graph}
\label{s:dataflow}
To partition a dataflow graph, one needs to specify which partition strategy to
use for each operator.  This section describes how \name finds the best
partition plan for a dataflow graph.

Different plans result in different running time and per-worker memory
consumption, due to factors including communication, GPU kernel efficiency and
synchronization.  Finding the best plan is NP-hard for an
arbitrary dataflow graph~\cite{spartan}.  Recent work has proposed an algorithm 
based on dynamic programming (DP) for partitioning a certain type of graphs.
\secref{ss:dataflow-coarsen} presents techniques to make
a dataflow graph applicable to DP, and \secref{ss:dataflow-recursive} improves
search time via recursion.

%
%


{\bf Optimization goal.} \REV{revision plan Sec.5 bullet 2}
\textred{Ideally, our optimization goal should consider both the end-to-end execution time of the partitioned dataflow graph and the per-worker memory consumption.  Unfortunately, neither metric 
can be optimized perfectly.  Prior work~\cite{jia2018beyond} optimizes the approximate
end-to-end execution time by minimizing the sum of total GPU kernel execution
time and total data transfer time.  

In \name, we choose to minimize the total communication cost based on two
observations.  First, the GPU kernels for very large DNN models
process large tensors and thus have similar execution time no matter which
dimension its input/output tensors are partitioned on.  Consequently, a
partition plan with lower communication cost tends to result in lower
end-to-end execution time.
Second, the memory consumed at each GPU worker is used in two areas: (1) 
for storing a worker's share of tensor data, (2) for buffering data for communication
between GPUs. The memory consumed for (1) is the same for every partition plan:
for $k$ GPUs, it is always $1/k$ times the memory required to run the dataflow
graph on one GPU. The memory consumed for (2) is proportional to the amount of
communication.  Therefore, a partition plan with lower communication cost
results in a smaller per-worker memory footprint. 
}

\subsection{Graph coarsening}
\label{ss:dataflow-coarsen}
\begin{figure}
	\centering
	\includegraphics[scale=0.60]{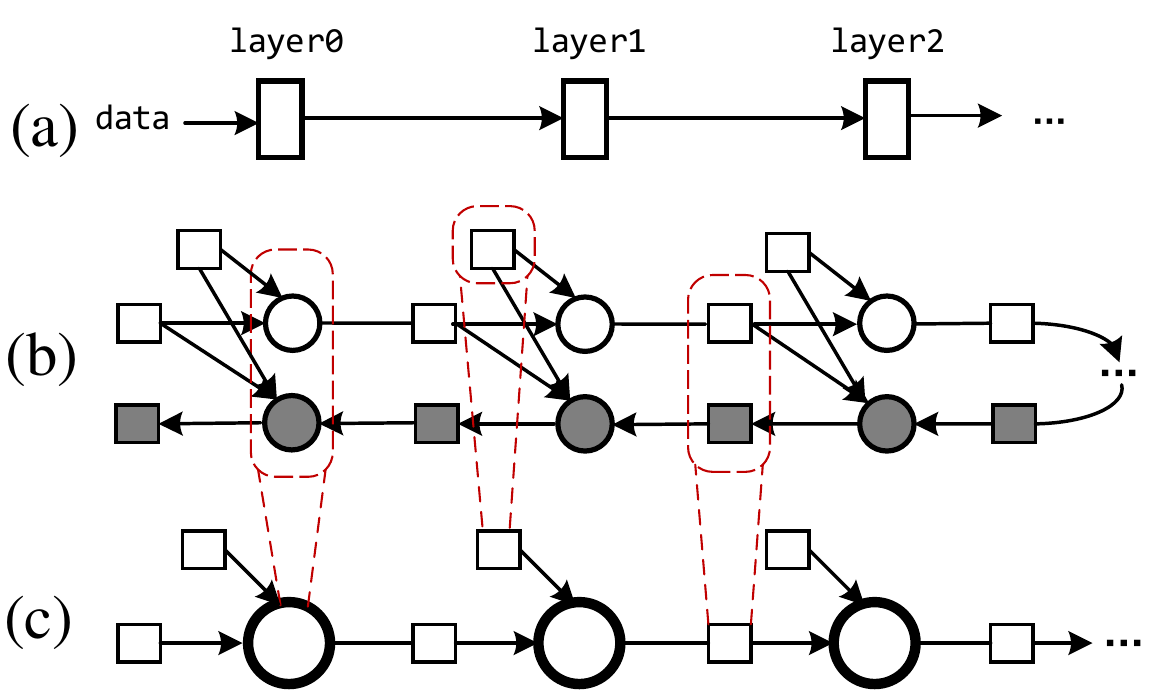}
	\caption{(a) Layer graph of a MLP model. (b) Its dataflow graph including forward
		and backward computation (in grey).
		(c) Coarsened graph. For cleanness, we only illustrate one operator group, one group for activation tensors and one group for weight tensor (dashed lines).}\label{fig:dataflow-fold}
\end{figure}

The algorithm in~\cite{jia2018exploring} is only applicable for linear graphs\footnote{We say a
graph $G$ is linear if it is homeomorphic to a chain graph $G'$, meaning there
exists a graph isomorphism from some subdivision of $G$ to some subdivision of
$G'$~\cite{homeomorphic}. Note that a ``fork-join'' style graph is linear by
this definition.}, such as 
the graph of DNN layers shown in Figure~\ref{fig:dataflow-fold}(a).  Dataflow
graphs of fine-grained operators are usually non-linear.  For example,
Figure~\ref{fig:dataflow-fold}(b) is the non-linear dataflow graph of the same DNN
represented by Figure~\ref{fig:dataflow-fold}(a).  Here, we propose to
``coarsen'' a dataflow graph into a linear one by grouping or coalescing
multiple operators or tensors.  

%

{\bf Grouping forward and backward operations.} Almost all DNN
models are trained using gradient-based optimization method. The training
includes a user-written forward propagation phase to compute the loss function
and a system-generated backward propagation phase to compute the gradients
using the chain rule. Thus, we coarsen as follows:
\begin{itemize}[leftmargin=*]
	\setlength\itemsep{0em}
	\item Each forward operator (introduced by the user) and its auto-generated backward operators
	(could be more than one) to form a group.
	\item Each forward tensor (e.g. weight or intermediate
	tensors) and its gradient tensor form a group.
		If a (weight) tensor is used by multiple operators during 
	forward propagation and thus has multiple gradients
	computed during backward propagation, the chain rule requires them to be summed
	up and the summation operator is added to the group as well.
\end{itemize}
Figure~\ref{fig:dataflow-fold}(c) shows the coarsened dataflow graph
for a MLP model. As forward and backward operators for the same layer are grouped
together,  the resulting graph becomes isomorphic to the forward dataflow
graph. For MLPs and CNNs, their coarsened graphs become linear. 
We perform the DP-based algorithm~\cite{jia2018exploring} on the coarsened graph. When the algorithm 
adds a group in its next DP step, we perform a brute-force combinatorial
search among all member operators/tensors within the group to find the minimal
cost for adding the group.  \REV{revision plan Sec.5, bullet 1}\textred{This allows tensors 
involved in the forward and backward
operators to be partitioned differently, while \cite{jia2018exploring} 
forces them to share the same partition configurations.} 
As there are only a few operators (typically 2) in
each group, the cost of combinatorial search is very low.

{\bf Coalescing operators.}
In DNN training, it makes sense for some operators to share the same partition
strategy.  These operators can be merged into one in the coarsened dataflow graph.
There are two cases:
\begin{itemize}[leftmargin=*]
\setlength\itemsep{0em}
\item {\em Merging consecutive element-wise operators}, because the input and output tensors
of an element-wise operator should always be partitioned identically.
We analyze the TDL description to determine if an operator is element-wise.
Consecutive element-wise operators are very common in DNN training. For instance,
almost all gradient-based optimizers (e.g. SGD, Adam, etc.) are composed of only element-wise operators.
\item {\em Merging unrolled timesteps}. Recurrent neural networks (RNNs) 
process a variable sequence of token over multiple timesteps. RNN has the key property that different time steps
share the same computation logic and weight tensors. Thus, they should be coalesced to share the same partition strategy. 
As a result, the dataflow graph of a multi-layer RNN becomes a chain of
coalesced and grouped operators. To detect operators that belong to different timesteps of the same
computation, we utilize how RNN is programmed in DNN frameworks. 
For example, systems like MXNet and PyTorch
call a built-in function to unroll a basic unit of RNN computation into many timesteps, allowing 
\name to detect and merge timesteps.
\end{itemize}

%
%
%

\subsection{Recursive partitioning}
\label{ss:dataflow-recursive}

\begin{figure*}[t]
	\centering
	\includegraphics[scale=0.53]{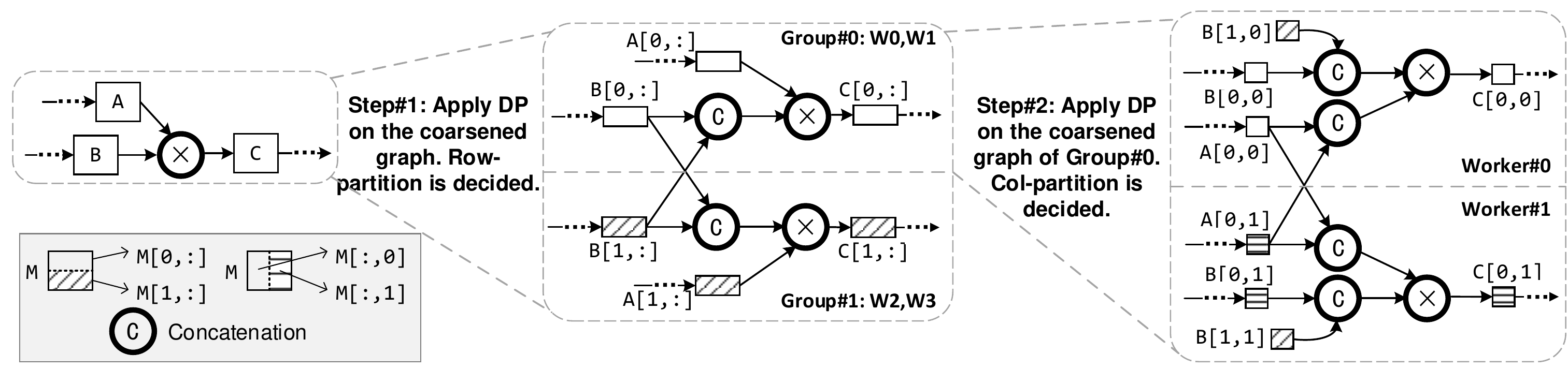}
	\caption{Recursively partition a dataflow graph to four workers.
		Only one matrix	multiplication is drawn for cleanness. In step\#1, every matrix is partitioned
		by row, and for group\#0, \code{B[1,:]} is fetched from the other group. Because of this,
		\code{B[1,:]} becomes an extra input in step\#2 when the graph is further partitioned to two
		workers. Because step\#2 decides to partition every matrix by column, every matrix is partitioned
		into a 2x2 grid, with each worker computes one block.}
	\label{fig:dataflow-gen}
\end{figure*}

\begin{table}[t]
	\centering
	\footnotesize
	\begin{tabularx}{\columnwidth}{|>{\centering}X|X|X|}
		\hline
		 & \multicolumn{2}{c|}{\textbf{Search Time}} \\
		\cline{2-3}
		 & WResNet-152 & RNN-10 \\
		\hline
		Original DP~\cite{jia2018exploring} & n/a & n/a \\
		\hline
		DP with coarsening & 8 hours & >24 hours \\
		\hline
		Using recursion & 8.3 seconds & 66.6 seconds \\
		\hline
	\end{tabularx}
	\caption{Time to search for the best partition for 8 workers. WRestNet-152 and RNN-10 are two large DNN models described in \secref{sec:eval}.}
	\label{tbl:improve}
	\vspace{-10pt}
\end{table}

When there are more than two workers, each operator can be partitioned along multiple dimensions.
This drastically increases the number of partition strategies available to each
operator and explodes the running time of the DP-based search algorithm.
To see this, consider the
coarsened graph of Figure~\ref{fig:dataflow-fold}(b). Every operator group
has two input tensor groups and one output tensor group. Each tensor group contains
one forward tensor and one gradient tensor. At each step, the DP
algorithm needs to consider all the possible configurations of an operator
group including different ways to partition the six input/output tensors. 
For each 4D tensor used in 2D-convolution, there are in total
20 different ways to partition it evenly across 8 workers. Hence, the number of
possible configurations of 2D-convolution's operator group 
is $20^6=6.4\times 10^7$. Although not all the dimensions are available
for partition in practice (e.g. the convolution kernel dimension is usually very small)
, the massive search space still results in $8$ hours of search time when partitioning the 
WResNet-152 model (Table~\ref{tbl:improve}). 

Our insight is that the basic DP search algorithm can be \emph{recursively} applied. For
instance, a matrix, after being first partitioned by row, can be partitioned again.
If the second partition is by column, the matrix is partitioned into a 2$\times$2
grid; if the second partition is by row, the matrix is partitioned into four parts
along the row dimension.

This observation inspires our recursive optimization
algorithm to handle $k=2^m$ GPUs:

\begin{enumerate}[leftmargin=*]
	\setlength\itemsep{0em}
\item Given a dataflow graph $G$, run the DP algorithm with coarsening 
to partition $G$ for two worker groups, each consisting of $2^{m-1}$ workers. 
Note that each tensor is only partitioned along \emph{one} dimension. 
\item Consider the partitioned dataflow graph as consisting of two halves:
$G_0$ for worker group\#0 and $G_1$ for worker group\#1. Each half also contains the data fetched
from the other group as extra input tensors.
\item Repeat step 1 on $G_0$ and apply the partition result to $G_1$ until there is only one worker per group.
\end{enumerate}
This recursive algorithm naturally supports partitioning along multiple dimensions.
Figure~\ref{fig:dataflow-gen} illustrates two recursive steps using an example dataflow graph (for brevity, we 
only show one matrix multiplication operator in the graph). Note the recursion must be 
done over the entire dataflow graph instead of a single operator, as the
partition plan of the previous recursive step will influence the global
decision of the current one. 

\textred{While the recursive algorithm may seems straightforward, it is less obvious why the 
resulting partition plan has the optimal overall
communication cost.  In particular, the recursive algorithm chooses a 
sequence of basic partition plans $\{\mathcal{P}_1, \mathcal{P}_2, ...
\mathcal{P}_m \}$ in $m$ recursive steps, and we need to prove that no other sequence of choices leads
to a better plan with a smaller communication cost.}
The main insight of our proof is that the partition plan decided in each recursive step 
is commutative (i.e,
choosing partition plan $\mathcal{P}$ followed by $\mathcal{P}'$ results in the same
total communication cost as choosing $\mathcal{P}'$ followed by $\mathcal{P}$.) Based on this insight, we
derive the following property and use it to prove optimality.

\begin{thm:greedy}
	Let the total communication cost incurred by all worker groups at step $i$ be $\delta_i$. Then $\delta_i\leq \delta_{i+1}$.
\end{thm:greedy}
Suppose 
$\{\mathcal{P}_1, \mathcal{P}_2, ... \mathcal{P}_m \}$ is the sequence of partition plans chosen and it is not optimal. 
Then there exists a different sequence
$\{\mathcal{P}'_1, \mathcal{P}'_2, ... \mathcal{P}'_m \}$
with smaller total cost.  Hence, there must
be two consecutive steps $k-1$ and $k$, such that $\delta_{k-1}\leq\delta'_{k-1}$ and
$\delta'_{k}<\delta_{k}$.  We can show that, by choosing $\mathcal{P}'_k$
instead of $\mathcal{P}_k$ at step $k$, the search could have produced a better
partition plan.  This contradicts the optimality of the DP algorithm. See appendix
for the full proof.




If the number of GPUs $k$ is not a power of two, we factorize it to $k = k_1*k_2*...*k_m$, where 
$k_i\geq k_{i+1}$ for all $i$.  At each step $i$ in the recursive algorithm, 
we partition the
dataflow graph into $k_i$ workers in which
each partition strategy still partitions a tensor along only one dimension but across $k_i$
workers. 


\vspace{0.1in}
{\bf The benefits of recursion.}
Recursion dramatically cuts down the search time by partitioning along only one
dimension at each step. For example, the number of configurations
to be enumerated at each step for a 2D-convolution operator group is only
$4^6=4096$.  Therefore, the total number of partition strategies searched for the 
2D-convolution operator with 8 workers (3 recursive
steps) is $3*4096$, which is far fewer than $20^6$ when recursion is not used.
Table~\ref{tbl:improve} shows the search time for two common large DNN models
when applying the original DP algorithm on coarsened graph without and with
recursion.  

As another important benefit, recursion finds partition plans that 
work well with common hierarchical physical interconnects which have 
less aggregate bandwidth near the top of the hierarchy.  For example, many
commercial servers group GPUs by faster PCI-e buses first and then connect the
groups with slower QPI buses or Infinibands. As theorem~\ref{thm:greedy}
indicates, \name assigns worker groups with less communication near the top of
the hierarchical interconnects in earlier steps of the recursion.

\section{Optimizations in generating the partitioned graph}
\label{ss:dataflow-gen}

\begin{figure}[t]
	\centering
	\includegraphics[scale=0.8]{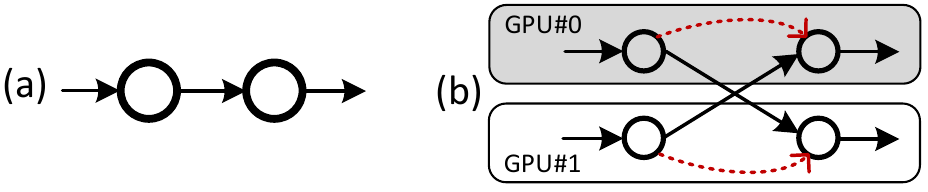}
	\caption{(a) Original dataflow graph; (b) Partitioned graph with extra control dependencies (dashed lines).}\label{fig:rewire}
\end{figure}

Once the search algorithm determines how to partition for every tensor and
operator, \name generates a corresponding partitioned dataflow graph.  The
graph generation process is mostly straightforward save for two optimizations,
which are crucial to keep the per-worker memory consumption low.

\textbf{Leveraging the existing memory planner.} Systems like MXNet and TensorFlow 
have their own memory planners to statically allocate and re-use memory buffers
among operators according to their dependencies.  Ideally, the per-worker
memory consumption for $k$ workers should be $1/k$ of the original memory
consumption.  In our initial implementation, per-worker memory consumption 
far exceeded the expected amount.  We found that this is because 
the partitioning of a dataflow graph changes the dependencies
between original operators.  Figure~\ref{fig:rewire} illustrates an example. In
the original graph, the second operator can reuse the memory buffer of the
first one (such as the workspace of a convolution operator) due to the
dependency between the two.  Naive graph generation may result in the graph with solid edges in
Figure~\ref{fig:rewire}(b), in which the two operators executed by each worker no longer have 
a direct dependency between them and thus allows no immediate memory-reuse.  To fix
this, \name maintains the original operator dependencies on each worker by
generating the extra control dependencies (dashed lines), so that the memory
planner can immediately re-use buffers across dependent operators.

\textbf{Fusing operators for remote data fetch.}
For each operator in the original graph, \name generates a copy for each GPU
worker in the partitioned graph.  Often, these operators need to fetch data from
a different worker. MXNet already supports copy, split, concatenate operators,
which can be used to support data movements.  A naively generated graph would
use split to extract the required input regions from the other workers, copy
data to the local worker, and concatenate them together to assemble the input
region needed by the operator's GPU kernel.  Extra reduce operators can also be
generated if the output tensors of different workers need to be aggregated
according to the partition strategy used. Execution of such graphs results in
many intermediate memory blocks, increasing the per-worker memory consumption.
To mitigate this, we wrote a custom GPU kernel called MultiFetch to retrieve 
remote data and assemble the input region in-place using CUDA Unified Virtual
Addressing (UVA). \REV{revision plan Sec.6 bullet 2}\textred{CUDA UVA allows a kernel running 
on one GPU to directly access the memory on another, which avoids explicit data copying before kernel execution.  Our MultiFetch kernel takes multiple pointers to the memory blocks of
the input regions from the other GPUs and assembles them in one kernel launch}.

Beyond the two optimizations described above, we also spread out the reduction workload
to all GPUs (all-reduce) when performing output reduction.
This is important for avoiding any single aggregation bottleneck.
We also find that the MXNet scheduler can execute the remote fetch operator
much earlier than required, resulting in memory being occupied for longer than
necessary. We adopt the same technique proposed by TensorFlow to delay the
execution of the remote fetch operator.

\section{Evaluation}
\label{sec:eval}
\textred{
This section evaluates \name and compares with various alternative approaches.
The highlights of our results are the following:
\begin{itemize}[leftmargin=*]
	\setlength\itemsep{0em}
\item \name can train very large WResNet and RNN models across 8 GPUs 
with high throughput that is within 60\%-98\% of a hypothetical ideal baseline.
\item Except for a few exceptions, {\name} outperforms existing alternative approaches including shrinking the mini-batch size used for training, swapping to CPU memory, and placing different operators on different GPUs.
\item \name's recursive partition algorithm leads to 
better training throughput than existing partition algorithms~\cite{jia2018exploring, spartan} and simple 
heuristics.
\item The overall partition plan found by {\name} is highly non-trivial, even though the underlying DNN 
model has a regular structure.
\end{itemize}
}
\subsection{Experimental setup}
\begin{figure*}[!pt]
	\centering
	\begin{subfigure}[t]{0.32\linewidth}
		\centering
		\includegraphics[width=\linewidth]{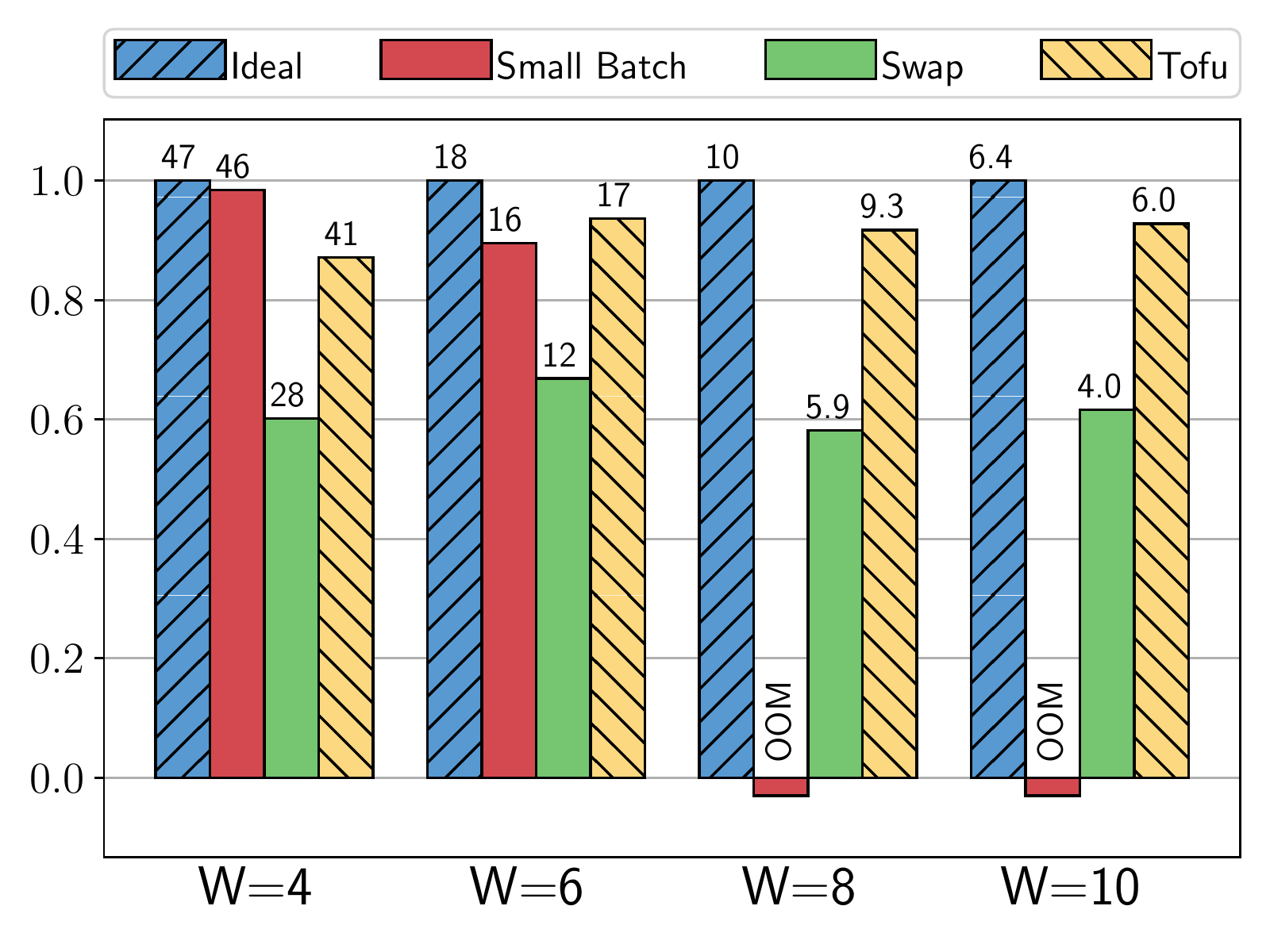}
		\caption{Wide ResNet-50}\label{fig:resnet_50}
	\end{subfigure}
	\begin{subfigure}[t]{0.32\linewidth}
		\centering
		\includegraphics[width=\linewidth]{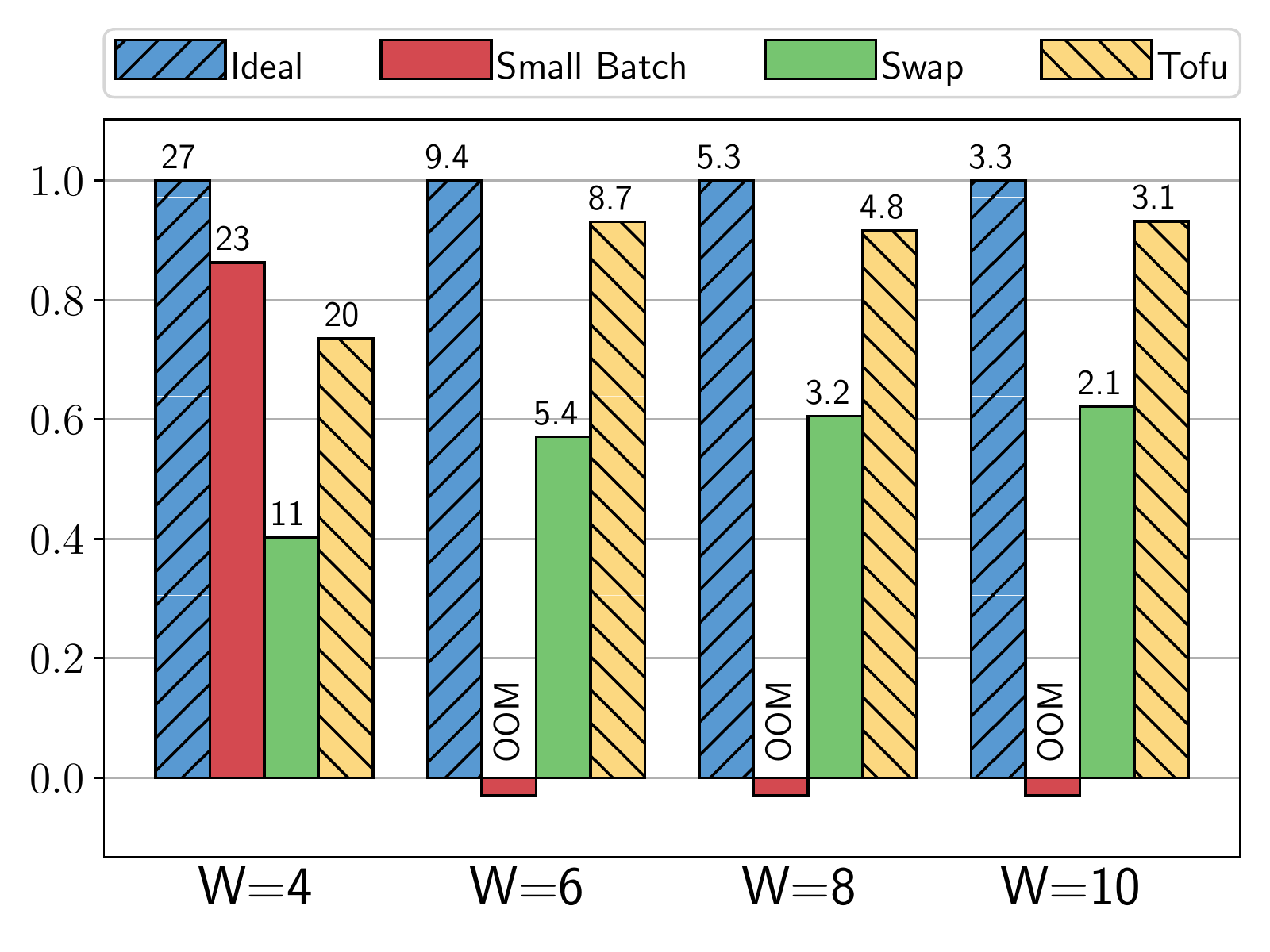}
		\caption{Wide ResNet-101}\label{fig:resnet_101}
	\end{subfigure}
	\begin{subfigure}[t]{0.32\linewidth}
		\centering
		\includegraphics[width=\linewidth]{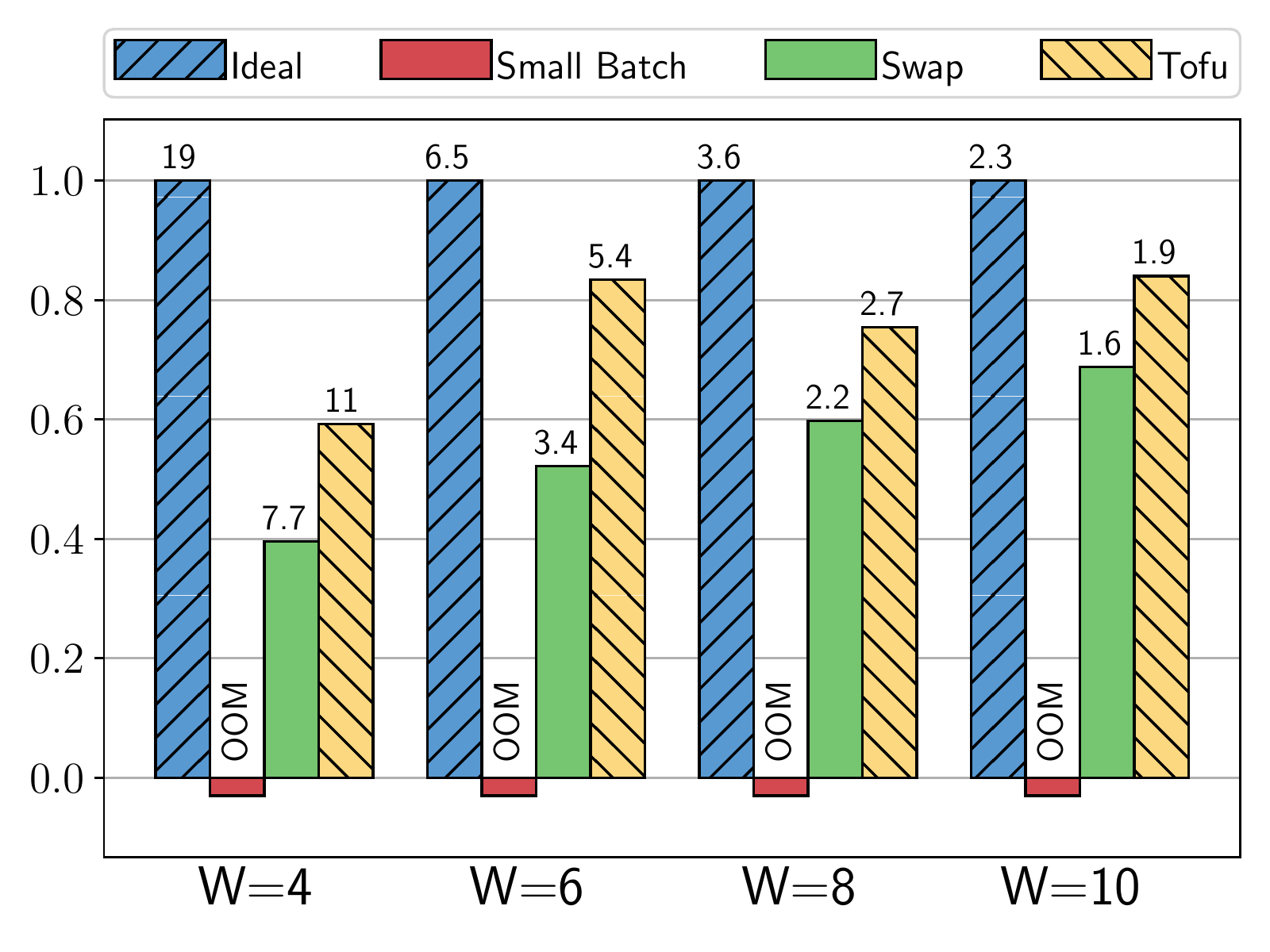}
        \caption{Wide ResNet-152}\label{fig:resnet_152}
	\end{subfigure}
	\captionsetup{aboveskip=0pt}
    \caption{Normalized \wnet throughput relative to the ideal performance.
             The number on each bar shows the absolute throughput
             in samples/sec.}
    \label{fig:resnet_throughput}
\end{figure*}


\begin{table}[t]
	\centering
	\footnotesize
	\begin{subtable}[t]{0.49\columnwidth}
	\begin{tabularx}{\linewidth}[t]{|X|X|X|X|}
		\hline
		\multicolumn{4}{|c|}{\textbf{RNN}} \\
		\hline\hline
		 & L=6 & L=8 & L=10 \\
		\hline
		H=4K             & 8.4 & 11.4 & 14.4 \\
		\hline
		H=6K             & 18.6 & 28.5 & 32.1 \\
		\hline
		H=8K             & 33.0 & 45.3 & 57.0 \\
		\hline
	\end{tabularx}
	\end{subtable}
	\begin{subtable}[t]{0.49\columnwidth}
	\begin{tabularx}{\linewidth}[t]{|X|X|X|X|}
		\hline
		\multicolumn{4}{|c|}{\textbf{Wide ResNet}} \\
		\hline\hline
		 & L=50 & L=101 & L=152 \\
		\hline
		W=4             & 4.2 & 7.8 & 10.5 \\
		\hline
		W=6             & 9.6 & 17.1 & 23.4 \\
		\hline
		W=8             & 17.1 & 30.6 & 41.7 \\
		\hline
		W=10            & 26.7 & 47.7 & 65.1 \\
		\hline
	\end{tabularx}
	\end{subtable}
	\caption{\small{Total weight tensor sizes (GB) of our benchmarks.}}
	\label{tbl:memory}
	\vspace{-10pt}
\end{table}

\noindent\textbf{Prototype Implementation.}
We implement \name based on MXNet 0.11.
The TDL components (operator
descriptions and the region analyzer) are written in Python (2K LoC).
The recursive search algorithm is implemented as a graph transformation
pass in NNVM (4K LoC in C++).  As we need information from
gradient calculation and shape inference, we also made slight modifications to
the corresponding NNVM passes.

\noindent\textbf{Testbed:} The experiments run on an EC2 p2.8xlarge instance.
The instance has 8 K80 GPUs with 12GB memory each.
GPUs are connected by PCI-e bus with 21GB/s peer-to-peer bandwidth.
It has 32 virtual CPU cores and 488GB CPU memory. The CPU-GPU bandwidth is 10GB/s.

\noindent\textbf{DNN Benchmarks:} We evaluate the WResNet~\cite{wideresnet}
convolutional neural network and recurrent neural network (RNN).  We choose
these two benchmarks because they correspond to very large models.
We do not evaluate those well-known DNNs that fit into a single GPU's memory,
such as AlexNet, VGGNet and Inception.

WResNet~\cite{wideresnet} is a widened version of the original residual network
model~\cite{resnet}. It has a widening scalar to increase the number
of channels on each convolution layer. The model size 
grows quadratically as each weight tensor is widened on both the input and
output channel. WResNet has been shown to achieve a better accuracy when the
model is widened by $10\times$. Due to the memory limitation, such improvement
is only demonstrated on CIFAR-10 dataset of small images (32x32) using a
50-layer model. We experiment with WResNet on ImageNet dataset with
images of size (224x224). We also test different model variations: widening
scalar from 4 to 10 on networks with 50, 101 and 152 layers.  We use notations
like WResNet-101-8 to denote the 101-layer ResNet model widened by 8 times.

For RNN, there are two ways to increase model capacity. The number of neurons in each
hidden layers can be increased, and multiple RNN layers
can be stacked to form a deeper model. Researchers have explored very large
RNNs by increasing the number of RNN layers to 8~\cite{mirhoseini2017device,
placement:iclr18}, or by using a large hidden layer size such as
8192~\cite{rnnlimit}.  We use the model described in~\cite{rnnlimit}, and test
it with different configurations varying from 6 to 10 layers with 4K, 6K, and
8K hidden sizes. All RNN model variants use LSTM cell~\cite{lstm} and are
unrolled for 20 steps as in~\cite{rnnlimit}. We use the RNN-8-8K to denote the
8-layer RNN model with 8K hidden size.

All the benchmarks are tested by running a full training iteration including forward/backward
propagation and weight update. State-of-the-art weight optimizers such as
Adam~\cite{adamoptimizer} and Adagrad~\cite{adagrad} must maintain an
extra buffer for storing the gradient history. Therefore, a model of weight size
$W$ needs to consume at least $3W$ size of memory for storing the weight, gradient
and the history tensors.  Table~\ref{tbl:memory} shows the total weight memory
consumption for all the benchmarks.


\vspace{0.1in}
\noindent\textbf{Baseline and Alternatives for Comparison.} 
We consider an ideal baseline and several alternative approaches for comparison.

\noindent\emph{Ideal} is a hypothetical baseline that assumes each GPU has
infinite memory. We simulate this by modifying the memory allocator of MXNet to
always return the same memory block.  We measure the single-GPU
throughput number and multiply it by 8 as the performance of running on 8 GPUs. 

\noindent\emph{SmallBatch} is a baseline that tries to fit the model in a
single GPU by reducing the mini-batch size. Like the {\em ideal} baseline, we scale the single-GPU
throughput number by 8 for 8 GPUs. Thus, neither SmallBatch nor Ideal baseline 
consider the communication cost and represent performance upper-bounds.

\noindent\emph{Swapping}~\cite{meng2017training,sekiyama2018profile,rhu2016vdnn}
is a baseline that swaps in/out GPU memory blocks to CPU.
There are many ways to design the swapping policy. Our baseline combines
many of these techniques in order for a fair comparison. First, our baseline
follows the design of~\cite{rhu2016vdnn}, which includes
a least recently used algorithm to decide the tensor to be swapped out
and a prefetching unit based on the execution. This supports swapping in/out
any memory block instead of only activation tensors as in~\cite{meng2017training}.
Second, read-only tensors are copied to CPU only once and simply dropped the next
time they are to be swapped out.
Third, we combine dataflow analysis similar to~\cite{meng2017training} to disable
swapping out memory blocks that will soon be used.

\noindent\emph{Operator Placement}
\cite{sutskever2014sequence,mirhoseini2017device,shazeer2017outrageously,google:nmt}
assigns operators to different devices to spread out memory usage.
For RNN, this baseline assigns the computation of different layers
to different GPUs to leverage the pipelining effect,
as it is originally proposed in \cite{sutskever2014sequence}. If there are more
layers than the number of GPUs, we balance the assignment in a round-robin manner.
Operator placement does not perform well for CNNs due the mostly serial
layer-by-layer execution. Therefore, we skip this baseline for all \wnet
benchmarks.

In our experiments, the ideal baseline uses a batch size that can saturate the
GPU for the best performance. SmallBatch, Swapping and Tofu all use the largest batch size
that make the execution fit in the GPU memory.


\subsection{Training Large and Deep Models}
\label{sec:throughput}

We show the performance of \name and compare it to the ideal baseline and alternatives.
\textred{Since different systems use different batch sizes to achieve
the best performance, we use throughput (samples/sec) instead of training time per iteration 
as the metric for comparison
}\REV{reviewer E, paragraph 3}
In Figures~\ref{fig:resnet_throughput} and~\ref{fig:rnn_throughput},
each bar shows the throughput relative to the ideal baseline performance.
The absolute throughput numbers are shown on top of
each bar. {\em OOM} indicates out-of-memory error.

{\bf \wnet Performance.}
Figure~\ref{fig:resnet_throughput} shows the \wnet throughput achieved by
different systems.  The ideal baseline uses a global batch size of 128. Only 3
models, \wnet-50-4,6 and \wnet-101-4 can be fit in a single GPU memory by
shrinking the batch size (aka SmallBatch).

\name can achieve 60\%-95\% of the ideal performance for all the models.
\textred{The largest model, \wnet-152, has the biggest performance gap. This is 
because we configured the ideal baseline to use a much larger mini-batch size for 
peak throughput without any consideration for memory consumption.
For example, the ideal baseline uses base size
128 for \wnet-152-4 while \name can fit at most 32.}\REV{revision plan Sec.7 bullet 2}
The batch sizes used by \name ranges from 8 (for \wnet-152-10) to 128 (for \wnet-50-4).
\name performs better than alternatives 
in all scenarios except for \wnet-50-4 and \wnet-101-4, in which 
SmallBatch achieves 12\% and 15\% better throughput than \name.  This is
because convolution kernels have good GPU utilization even for small batch
sizes.  However, SmallBatch runs out of memory for most of the models
in Figure~\ref{fig:resnet_throughput}. 

As shown in Figure~\ref{fig:resnet_throughput}, swapping is 20\%-63\% slower than \name across all the models.  This is due to swapping's much larger communication amount.
Although we implemented prefetching to ``hide'' communication latency in swapping,
the CPU-GPU communication is the bottleneck as all 8 GPUs share the same bandwidth to
communicate with the CPU. 


\begin{figure*}[t!p]
	\centering
	\begin{subfigure}[t]{0.32\textwidth}
		\centering
		\includegraphics[width=\linewidth]{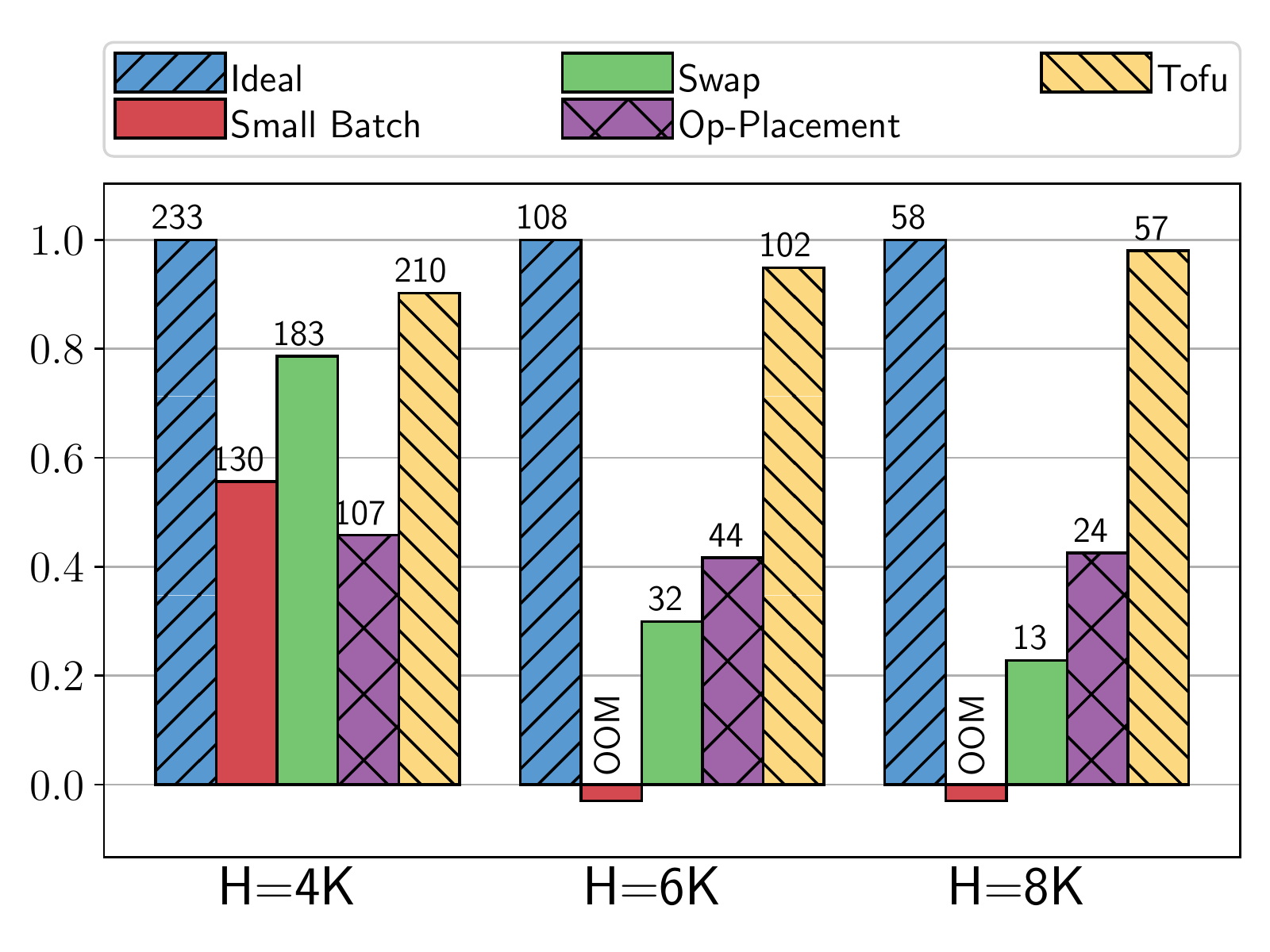}
		\caption{6 layers RNN}\label{fig:rnn_6}
	\end{subfigure}
	\begin{subfigure}[t]{0.32\textwidth}
		\centering
		\includegraphics[width=\linewidth]{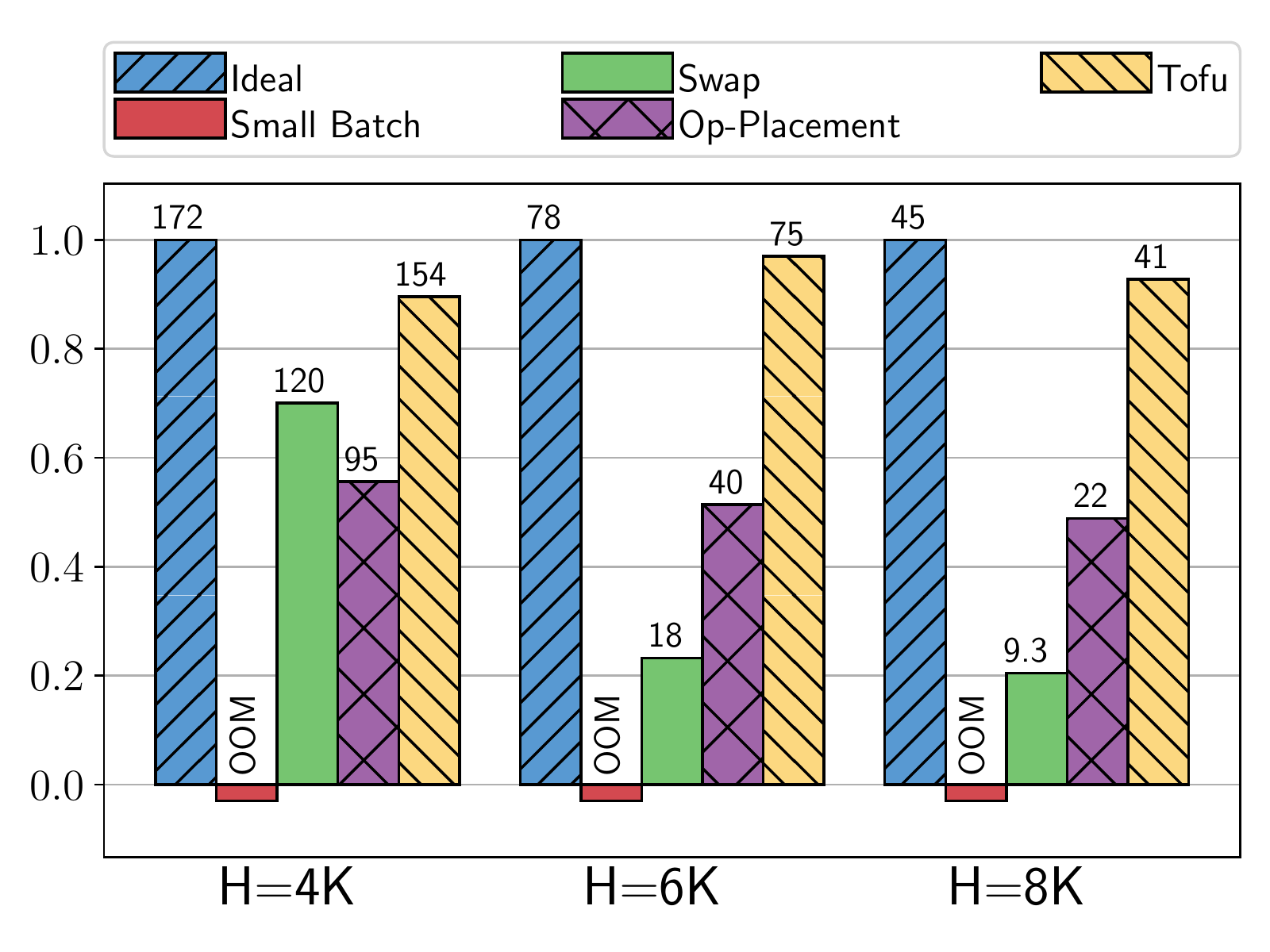}
        \caption{8 layers RNN}\label{fig:rnn_8}
	\end{subfigure}
	\begin{subfigure}[t]{0.32\textwidth}
		\centering
		\includegraphics[width=\linewidth]{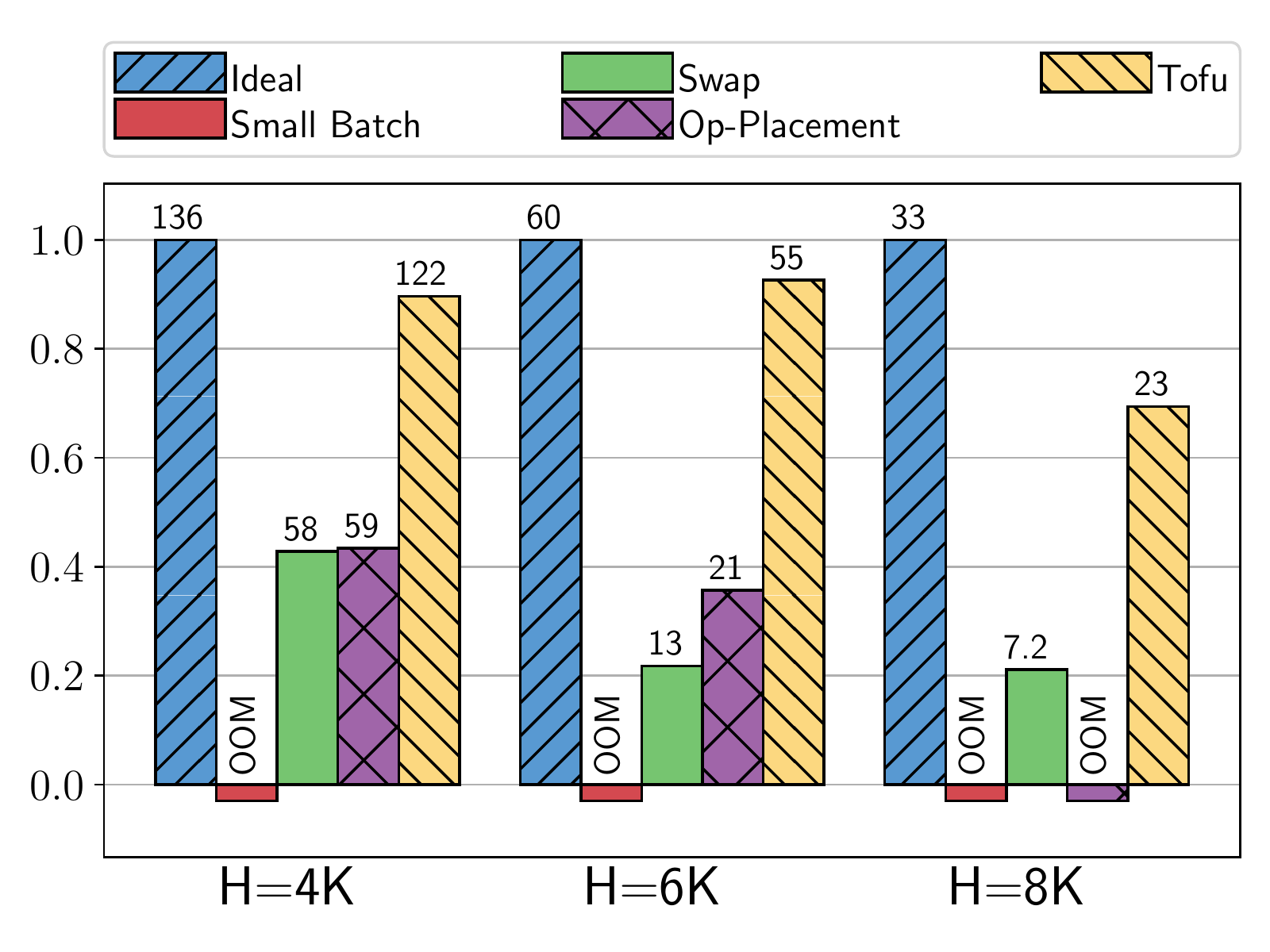}
        \caption{10 layers RNN}\label{fig:rnn_10}
	\end{subfigure}
	\captionsetup{aboveskip=0pt}
    \caption{Normalized RNN throughput relative to the ideal performance.
             The number on each bar shows the absolute 
             throughput in samples/sec.}
    \label{fig:rnn_throughput}
\end{figure*}

{\bf RNN Performance.}
Figure~\ref{fig:rnn_throughput} shows the throughput for RNNs.  The ideal
baseline uses a (global) batch size of 512. 
\name performs better than the other baselines in all RNN configurations,
achieving 70\% - 98\% of ideal throughput.  Unlike the \wnet experiments,
SmallBatch does not achieve better throughput than \name in any RNN
configuration. This is because the main RNN computation is matrix
multiplication, which has much less arithmetic density than convolution.  Thus,
performing matrix multiplication using small batch sizes results in decreased
GPU utilization.  The same reasoning explains why \name's relative performance
with the largest model (RNN-10-8K) is worse than with other RNN models; \name
uses a batch size of 128 in order to fit RNN-10-8K in memory while it uses
larger batch sizes (256 or 512) with other RNN models.
As is also the case with \wnet, SmallBatch results in OOM for larger RNN
configurations.

Operator placement achieves 38\%-61\% of \name's throughput and cannot train
RNN-10-8K (OOM). Two reasons contribute to the lower performance. First, layer-wise
placement results in imbalanced load because the number of layers is not a
multiple of the number of GPUs.  Second, layer-wise placement relies on
pipelined parallelism: GPU-1 executes the first operator in the first
layer and forwards its result to GPU-2. GPU-2 can execute the
first operator in the second layer while GPU-1 concurrently executes
the second node in the first layer.  Pipelined parallelism cannot fully
saturate GPUs at all times: e.g. GPU-2 is idle while GPU-1
executes its first operator.  By contrast, \name parallelizes the execution of
each operator and keeps all GPUs busy at all times.

Swapping achieves 23\% - 30\% throughput of \name and 48\% - 53\% throughput of
operator placement when the weight size is large.  The main reason is that many
tensors may be used simultaneously in RNN training.  To fully saturate a GPU,
most deep learning frameworks, including MXNet and Tensorflow, schedule operators
immediately when they are ready.  RNN's mesh-like dataflow graph results in
more tensors to be used at the same time. When the weight size is large, the
amount of swapping increases significantly.  Coupled with the CPU-GPU
communication bottleneck, swapping is unable to achieve good throughputs for
RNNs.

{\bf Comparing with TensorFlow.}
We compare with Tensorflow v1.8 (using Op-Placement) for training RNNs.
Table~\ref{tbl:tensorflow} shows the throughputs for running on RNN-6-4K, RNN-8-4K,
and RNN-10-4K. For additional comparison points, we also include MXNet (using Op-Placement).
Note that the throughputs of \name and MXNet are same as those in Figure~\ref{fig:rnn_throughput}.
Tensorflow's throughput is roughly half of MXNet and about 23\% of 
\name.  As Tensorflow and MXNet use the same operator kernel implementations, we originally expected 
the two systems to have similar throughput. However, further investigation shows that
TensorFlow \textred{does not support in-place gradient aggregation which may be crucial for 
the performance of large RNNs.}

\begin{table}[t]
    \centering
    \footnotesize
    \begin{tabularx}{\columnwidth}{|>{\hsize=.4\hsize}X|>{\hsize=.2\hsize}X|>{\hsize=.2\hsize}X|>{\hsize=.2\hsize}X|}
      \hline
                           & \textbf{RNN-6}   & \textbf{RNN-8}   & \textbf{RNN-10}\\
      \hline
       \name               & 210    & 154  & 122\\
      \hline
       MX-OpPlacement      & 107    & 95  & 59\\
      \hline
       TF-OpPlacement      & 50     & 36  & 30\\
      \hline
    \end{tabularx}
    \caption{Comparison of throughput (samples/second) for RNN models. The hidden size is 4096.}
    \vspace{-10pt}
\label{tbl:tensorflow}
\end{table}

\subsection{Comparing different partition algorithms}
\label{ss:eval-partition}
\begin{figure}
    \centering
    \begin{subfigure}[t]{0.49\columnwidth}
    	\centering
    	\includegraphics[width=\linewidth]{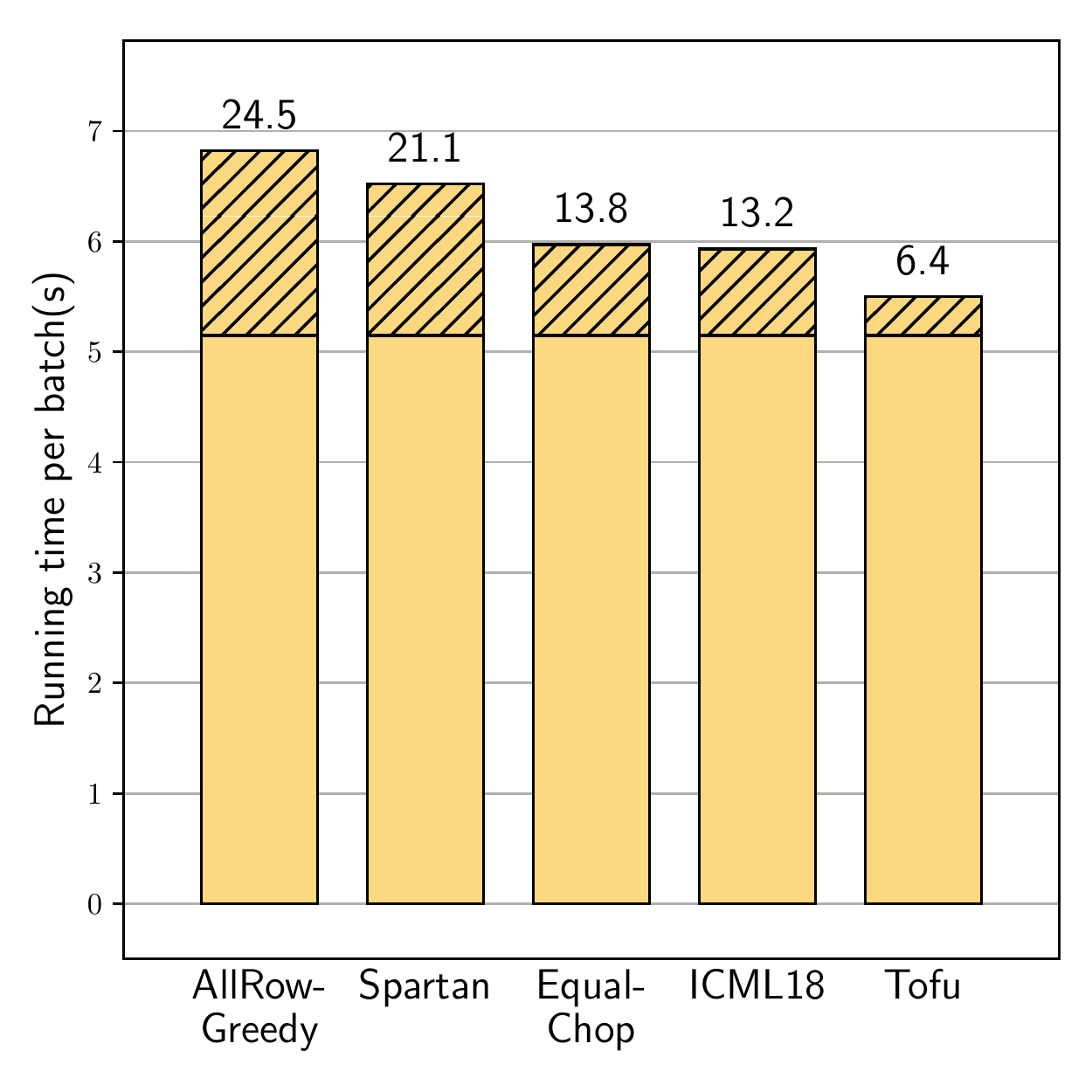}
        \caption{RNN-4-8K}
    \end{subfigure}
    \begin{subfigure}[t]{0.49\columnwidth}
     	\centering
       	\includegraphics[width=\linewidth]{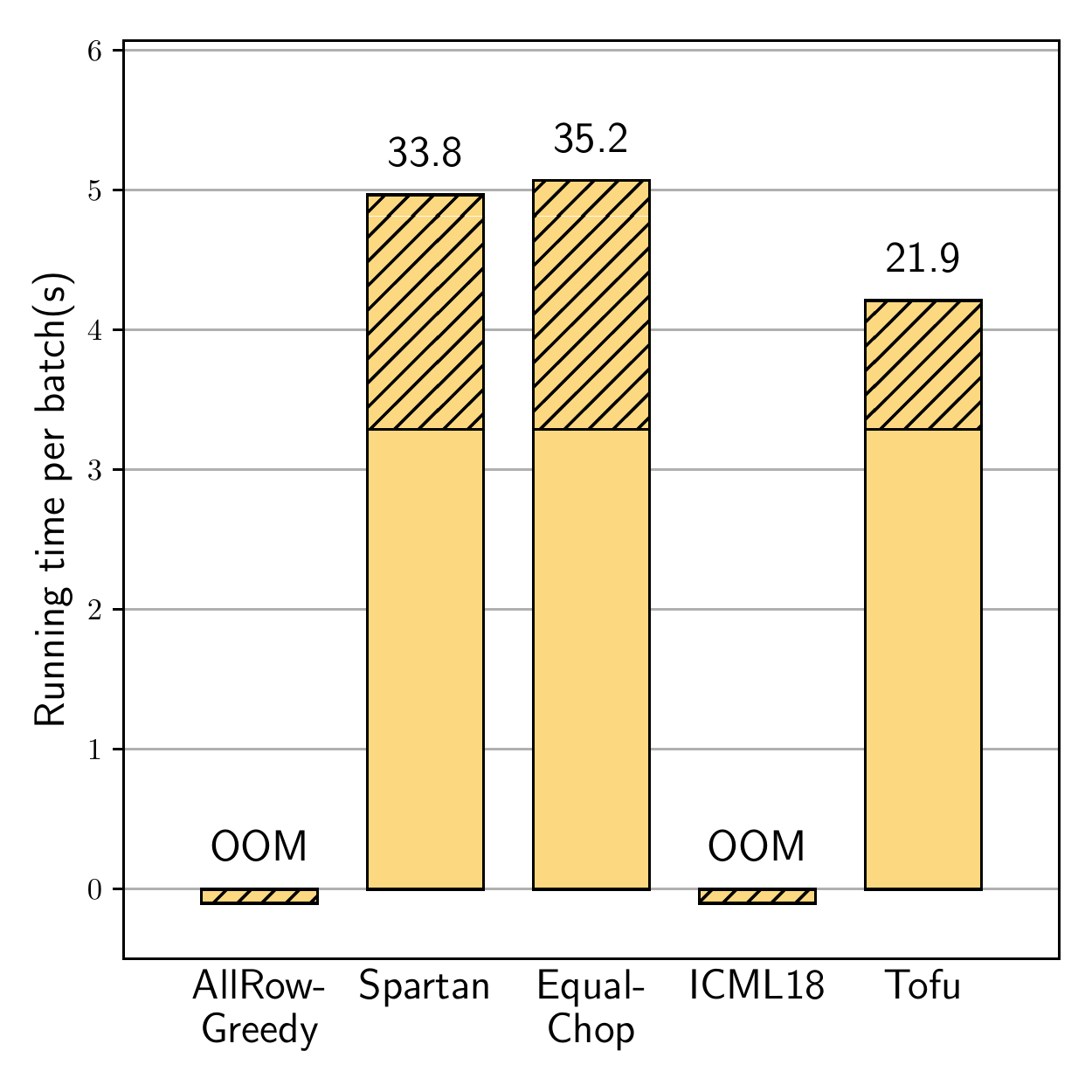}
        \caption{WResNet-152-10}
    \end{subfigure}
    \caption{Comparison of different partition algorithms using RNN-4-8K and
             WResNet-152-10 on 8 GPUs. Striped parts show the overhead
             (percentage) due to communication.}
    \label{fig:partition}
\end{figure}

\begin{figure*}[!pt]
    \centering
    \includegraphics[width=\linewidth]{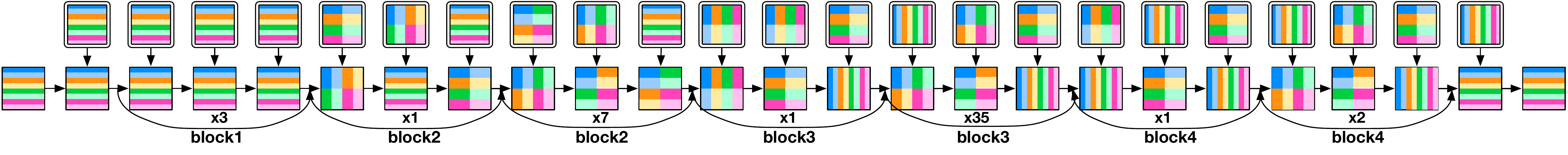}
    \caption{The partition found by \name for WResNet-152-10 on 8 GPUs.
    \textred{We draw the weight tensors (top row) and the activation/data
	tensors (bottom row) used by convolution operators. Partitioning is marked
	by the tiles and each color shows the tiles owned by the same GPU.}
	The vertical and horizontal dimensions of an activation tensor indicate the batch and channel
    dimensions. 'xN' symbol means the corresponding block is repeated N times.
    }
    \label{fig:resnet_partitions}
\end{figure*}

We have compared \name's search time with the original DP
algorithm~\cite{jia2018exploring} in \secref{ss:dataflow-recursive}
(Table~\ref{tbl:improve}). We now compare the quality of partition 
plan found by \name vs. ~\cite{jia2018exploring} and various other heuristics. 

The simplest heuristic (\texttt{\small
AllRow-Greedy}) partitions all tensors along the first dimension and partitions each
operator using the best strategy given that its input/output tensors are
partitioned on the first dimension. Note that, for the case of WResNet, this gives similar
result as the \emph{one-weird-trick} strategy proposed in~\cite{krizhevsky:wierd},
because all the convolution layers are partitioned by the batch dimension
and the only fully-connected layer in WResNet occupies <1\% of the total time.
Our next heuristic is to greedily
partition the largest tensor first (along any dimension), followed by its
incident operators, followed by the second largest tensor and so on. This is
equivalent to what is proposed by Spartan~\cite{spartan}. We also compare with
\name's DP algorithm applied to chop each tensor equally along only one
dimension (\texttt{\small EqualChop}).  Finally, we compare with the algorithm
in~\cite{jia2018exploring}(\texttt{\small ICML18}) which does not consider the
partition strategy of aggregating output tensors (aka output-reduction). 

Figure~\ref{fig:partition} shows the execution time of training one batch on 8 GPUs for 
RNN-4-8K (batch size is 512) and WResNet-152-10 (batch size is 8).  
To see the impact of communication on the execution time, we modify the
backend to skip memory copy among GPUs and measure the resulting pure
computation time, which is shown as the lower light-colored portion of the bars in
Figure~\ref{fig:partition}.

\texttt{\small AllRow-Greedy} performs worse among all the algorithms and
run out of memory for WResNet-152-10 because it needs to fetch too much
data from the other GPUs.  \texttt{\small Spartan} and \texttt{\small EualChop}
reduce the communication overhead by 3\%-10\% but are still worse than \name.
This result shows the benefit of partitioning a tensor along multiple dimensions.
\texttt{\small ICML18} is 7\% slower than \name for RNN-4-8K and results in OOM for
WResNet-152-10 due to the lack of output-reduction.  After adding output-reduction,
\texttt{\small ICML18} can find the same strategy as \name, albeit with a much
longer search time (see Table~\ref{tbl:improve}).

\subsection{Partition Results}
\label{ss:partition-result}
Figure~\ref{fig:resnet_partitions} shows the partition found by \name for WResNet-152-10.
ResNet-152 contains 4 groups of residual blocks: 
each block includes 3 convolutions 
and is repeated 3, 8, 36, and 3 times for each group respectively. \REV{revision plan Sec.7 bullet 3}
\textred{The lower
residual blocks (those close to the input layer) have larger feature map but smaller
weight tensors while the higher ones are the opposite. 

We make the following observations:
\begin{itemize}[leftmargin=*]
	\setlength\itemsep{0em}
	\item \name partitions both the batch and channel dimensions and the resulting partition plan 
	is a complicated combination of different partition strategies.
	\item \name chooses different partition plans for different convolution layers within one residual block.
	Repeated residual blocks are partitioned in the same way except for the first
	block in the group which has a different configuration to shrink the initial input feature
	map size by half.
	\item As the
	activation tensors in lower layers are larger and the weight tensor smaller, \name chooses to 
	fetch weight tensors from remote GPUs to save communication.
	As the weight tensors are larger in the higher
	layers, \name switches to partition strategies that fetch the relatively smaller
	activation tensors.
\end{itemize}
}

\section{Related Work}
\label{sec:related}

\noindent{\bf Parallel DNN training.}
Many parallel strategies have been developed to speedup
DNN training. Some strategies such as the popular
data parallelism~\cite{muli:osdi14,xingbounded:atc14,bosen:socc15,geeps:eurosys16}
cannot be used for training very large models because the parameters
are replicated to each device.
Model parallelism spreads out the model parameters to multiple GPUs,
thus is suitable for training very large models. 
Early work\cite{coates2013deep,krizhevsky:wierd,dean:nn} parallelizes 
specific classes of DNN models, and is
limited in flexibility and generality.
Minerva\cite{minerva} and Strads\cite{strads:eurosys16} require users to
implement extra interfaces to partition model parameters while \name
requires no change to the user program.
Another approach is to assign different layers/operators
to different devices via heuristics~\cite{shazeer2017outrageously} or
stochastic search~\cite{sutskever2014sequence,mirhoseini2017device}.
However, operator placement only works well only when there are
sufficiently many concurrent operators, and thus is not suitable
for DNN models with a deep stack of layers.


\noindent{\bf Out-of-core DNN training.}
This includes recomputation on demand
~\cite{gruslys2016memory,martens2012training,chen2016training}
, swapping and prefetching from host memory~\cite{meng2017training, sekiyama2018profile, rhu2016vdnn}.
Recomputation is not viable for large weight tensors. Swapping with host memory reduces
the opportunity of co-locating
computation and data, and scales poorly when there are multiple GPUs.
None of them can efficiently utilize the aggregated memory capacity
of multiple cards as \name does. Moreover, \name can also be combined with these
techniques.

\noindent{\bf Model compression.} This includes 
network pruning~\cite{han2015learning,han2015deep} (which removes small weight
values), quantization\cite{gong2014compressing} and reduced
precision\cite{hubara2016binarized}. The compressed model can then be deployed
on mobile or edge devices or to speed up the inference.  However, these
approaches affect model accuracy while \name allows exploring very large models
without changing the model behavior.

\noindent{\bf Parallel tensor computing.}
There is a long history in developing efficient parallel systems for tensor computing.
The very first effort starts from developing low-level, optimized, parallel matrix/tensor
libraries~\cite{lapack,scalapack,Elemental,globalarrays,petsc-efficient}.
These libraries implement efficient parallel matrix algorithms~\cite{cannon,summa} and tensor
operations~\cite{tensor-contraction}. However, they have very limited programmability support and
adding new operators requires tremendous manual efforts.

Many frameworks or tools have been built to ease the programming
of parallel tensor computation. In the low-level, ZPL~\cite{zpl}, Chapel~\cite{chapel}
and Unified Parallel C~\cite{upcspec} are parallel language supports. In the higher-level,
systems such as~\cite{SciHadoop,googler,sparkr,kasen,dmll,spartan}
let users write programs in high-level primitives like map and reduce. MadLinq~\cite{madlinq}
and Presto~\cite{presto} let user describe operators using parallel loop
primitives. Users need to express parallelism using the proper combination of
these primitives. For example, implementing a parallel matrix multiplication needs
to call the \code{shuffle} primitive in Spartan~\cite{spartan} or the \code{Collect}
primitive in ~\cite{dmll}. However, these primitives are limited (e.g. it
is hard to express halo-exchange in convolution).
Distributed Halide~\cite{distributedhalide} lets user
describe the algorithm in their DSL and specifies how it is paralleled.
As there are usually multiple ways of partitioning
data and computation, the efficiency varies with different implementations.
Spartan~\cite{spartan} and Kasen~\cite{kasen} propose algorithm to automatically
optimize array/matrix
partitioning to reduce communication. \cite{dmll} further improves this by also
considering different parallel patterns via transformations of nested high-level
primitives.

More recent proposals aim to fully automate the whole stack -- user programs
are written in array language and the system can distribute the data and computation
automatically. There are several approaches. Cylops Tensor Framework~\cite{ctf}
and Tensor Contraction Engine~\cite{tce} are specialized systems
for automatically parallelizing tensor contraction.
Spartan tries to map Numpy operators to high-level map and reduce primitives and
then partitions them accordingly. Others tried to leverage the parallelism
among array operators. For example,
Pydron~\cite{pydron:osdi14} translates Python program into an internal dataflow graph
to parallelize independent loops.
\cite{sutskever2014sequence,mirhoseini2017device} tries to dispatch array operators
to different devices automatically based on the dataflow graph. However, they
are not suitable for DNN computation that is mostly sequential.
Compared with previous
systems, \name automatically discovers the \pnr parallel patterns of operators
using TDL description and optimizes partitioning for the entire dataflow graph.

\noindent\textbf{Data layout optimization.}
There have been extensive work on optimizing communication (aka remote memory access)
on the multiprocessor architecture (e.g. ~\cite{hudak1990compiler,
knobe1990data, philippsen1995automatic,
milosavljevic1999automatic, ramanujam1991compile,ramanujam1989methodology,
bau1995solving, d1989partitioning, huang1993communication,ju1992reduction,
lu2009data}) or the new hardware~\cite{neurocube,tetris,yang2016systematic}.
Since searching the optimal solution is NP-Complete~\cite{kennedy1998automatic, kremer1993np, li1990index,
li1991data}, heuristics are used in practice~\cite{li1991data,
philippsen1995automatic}. By contrast, \name analyzes the relatively simpler
operator description language instead of the source code, and exploits the DNN
computation structure for its optimization.

\section{Discussion, limitations, and future work}
\label{s:discussion}

{\bf Fundamental limitations.}  \name only supports
parallelization via \pnr, which restricts each worker to perform
a coarse-grained task identical to the original computation.  This
pattern is not applicable to all parallelizable computation (e.g.
Cholesky~\cite{madlinq}).  Furthermore, the \pnr parallel strategies
do not necessarily minimize communication, and 
do not take advantage of the underlying interconnect topology.  By contrast,
parallel algorithms developed for specific computation (e.g. matrix
multiplication~\cite{cannon,summa}, tensor contraction~\cite{ctf}) are
explicitly structured to minimize communication and exploit the interconnect
topology.

\paragraph{\bf Limitations of TDL.} TDL is a simple language without
control flow primitives and data-dependent indexing.  Furthermore, \name 
does not support sparse tensor operations due to load-imbalance, even though they can usually be described in
TDL. For certain operations, these limitations may be removed by supporting
data-dependent partitioning (e.g. as in parallel graph
computation~\cite{powergraph:osdi12}) or by sampling runtime information (e.g.
as in parallel range sort~\cite{mapreduce}).

\name does not verify that the operator implementation matches its TDL description.
Such verification is an open research problem even if the underlying implementation
is open sourced.  A more promising direction is to leverage recent 
operator code-generation tools such as TVM~\cite{tvm:osdi18} and TC~\cite{tc:fb}.
As TVM and TC are also based on Halide, our analysis techniques can be ported
to analyze operators implemented in these languages.

\paragraph{\bf Partition flexibility and hardware heterogeneity.}  \name always
partitions every operator and tensor across all workers.  For moderately sized
DNN models, partitioning across all workers lead to small GPU kernels that leave
a GPU unsaturated.  In such scenarios, it may be beneficial to leave certain
operators un-partitioned or partially partitioned among a subset of workers.
Furthermore, \name has no support for non-uniform partitioning when GPUs
have different computing and memory capacity.  Although \name's
search algorithm tries to accommodate bandwidth differences in a hierarchical
interconnect, it does not explicitly optimize communication according to the
interconnect topology.

Unfortunately, \name's recursive search cannot be extended to address the 
above limitations. This is because the underlying DP algorithm cannot optimally search different
device placement choices for un-partitioned, or non-uniformly-partitioned
operators.  Exploring stochastic search
mechanisms~\cite{mirhoseini2017device,placement:iclr18,jia2018beyond} is a direction of future work. 


\textred{
\section{Conclusion}

We present the \name system, which enables the training of very large
DNN models by partitioning a dataflow graph of tensors across multiple GPU
devices. To automate this process, \name infers each operator's valid partition
strategies by analyzing its semantics written in a simple description language
(TDL). \name uses a recursive search algorithm based on dynamic programming and
DNN-specific heuristics to find the best partition plan that minimizes 
communication for the entire dataflow graph.  
}
\section*{Acknowledgements}
This work is supported in part by the National Science Foundation under award
CNS-1816717, NVIDIA AI Lab (NVAIL) at NYU, and AWS cloud credits for research. Our
shepherd, Chris De Sa, and other anonymous reviewers have given helpful
feedback that improved this work.  We also thank Jeff Hammond for pointing us
to related work in the HPC community, esp. work on tensor contraction
engines.

\bibliographystyle{unsrt}

\appendix
\section{Recursive Partitioning Algorithm and its Correctness}
\subsection{Recursive partition plan}
We first formally define the partition plan of a dataflow graph. Given a dataflow
graph $G$, a partition plan $P$ consists of the choices of how each tensor is partitioned
and how each operator is paralleled. Note that
the tensor can be partitioned along multiple dimensions but the number of splits
should be equal to the number of GPUs.

Given $2^m$ GPUs, any partition plan for a dataflow graph can be realized by
a sequence of recursive steps, $\langle p_1, p_2,\ldots, p_m \rangle$, where each $p_i$
is a \emph{basic partition plan} that partitions tensors along only one dimension 
among two (groups of) workers. Note that after $i$ steps, there are $2^i$ identical 
sub-dataflow graphs whose tensors are $1/2^i$ the original size. So the $p_{i+1}$
basic partition plan is applied to all $2^i$ sub-graphs.

\name's recursive partition algorithm chooses a
sequence of partition plans $P=\langle p_1, p_2,\ldots, p_m \rangle$
in $m$ recursive steps
and we want to show that this sequence is no worse than the optimal sequence
$O=\langle o_1, o_2,\dots,o_m \rangle$.


\subsection{Region Analysis}
Recall in Sec~\ref{ss:region}, we use symbolic interval to
analyze the access pattern of an operator. Let $\mathcal{X}_1,\ldots,\mathcal{X}_n$
and $\mathcal{Y}_1,\ldots,\mathcal{Y}_k$ be the symbolic upper bound of
each output index and access range of each input dimension, respectively.
The analysis produces following affine transformation:
\begin{equation}
\begin{pmatrix}
\mathcal{Y}_1 \\
\vdots \\
\mathcal{Y}_k
\end{pmatrix}
=
\begin{pmatrix}
\alpha_{11} & \alpha_{12} & \dots \\
\vdots & \ddots & \\
\alpha_{k1} & & \alpha_{k(n+1)}
\end{pmatrix} 
\begin{pmatrix}
\mathcal{X}_1 \\
\vdots \\
\mathcal{X}_n \\
1
\end{pmatrix}
\end{equation}

Here, we consider a restricted form of affine transformation. We prove
that the recursive algorithm is optimal under the following assumptions.
Whether these assumptions are necessary or not requires further
study.

\paragraph{Assumption\#1.} Each output index is used to access only
one dimension for each input tensor. The same output index can be used in
multiple input tensors such as element-wise operators
\code{\kw{lambda} i : A[i] + B[i]}, but \code{\kw{lambda} i: A[i, i]} is
not considered. In practice, we do not encounter any such example
when investigating operators in MXNet and Tensorflow.

\paragraph{Assumption\#2.} We only consider input dimensions
that scale linearly with one output index (i.e, in the form
of $\mathcal{Y}_i=\alpha_i\mathcal{X}_j$). For example, we
ignore partitioning on the third
dimension of \code{data} in \code{conv1d} (Figure~\ref{fig:conv1dtdl}).
This restriction rules out the partition-n-reduce strategies such
as halo exchange in convolution, but still includes many others
such as partitioning on channel dimension.

Because of the above assumptions, one immediate corollary is as follow.
\begin{corollary}\label{thm:shape}
Consider an operator that has output shape
$\mathcal{X}_1\times\ldots\times\mathcal{X}_n$.
The shape of any of its input tensors can be represented as
$\beta_{1}\mathcal{X}_{\pi_{1}}\times\ldots\times\beta_{d}\mathcal{X}_{\pi_{d}}$,
where $d$ is the number of dimensions, $\beta_1\ldots\beta_d$ are constants
and $\pi$ is a permutation of $1\ldots n$.
\end{corollary}

\subsection{Communication cost}

\begin{lemma}\label{thm:prop}
Let $\mathcal{T}(G)$ deonte the set of all tensors in a dataflow graph $G$.
The communication cost incurred by a basic partition plan $p$ is
a weighted sum of the size of each tensor:
\[
\text{cost}(p)=\sum_{t\in{\mathcal{T}(G)}}\alpha_tS_t\triangleq\vec{\alpha_p}\cdot\vec{S}
\]
, where $\alpha_t$ is some constant and $S_t$ is the size of tensor $t$.
We can thus further write it as dot product of two vectors.
\end{lemma}
\begin{proof}

Communication happens in two situations:
\begin{itemize}[leftmargin=*]
	\setlength\itemsep{0em}
	\item The selected partition-n-reduce strategy requires input
	region that is not available locally.
	\item The selected partition-n-reduce strategy produces output
	region that is assigned to other devices.
\end{itemize}
Consider an operator whose output tensor $t$ has shape
$\mathcal{X}_1\times\ldots\times\mathcal{X}_n$
and the partition plan
$p$ chooses to partition the dimension $i$ into halves. For the first case,
by Corollary~\ref{thm:shape},
the communication required to fetch one of the input tensor is either:
\[
\frac{1}{2}\Pi_{j=1}^d\beta_j\mathcal{X}_{\pi_j}=\left(\frac{1}{2}\Pi_{j=1}^d\beta_j\right) S_t
\]
if $i$ is not included in $\pi_1\ldots\pi_d$ (i.e, the whole tensor is needed), or
\[
\frac{1}{4}\Pi_{j=1}^d\beta_j\mathcal{X}_{\pi_j}=\left(\frac{1}{4}\Pi_{j=1}^d\beta_j\right) S_t
\]
, otherwise.

The same analysis can be applied to the second case. Because the total
communication is the summation of the cost to fetch each input and output
tensor, the result is a weighted sum of each tensor size.
\end{proof}

Let $\text{cost}(P)$ be the total communication cost of a partition plan sequence $P$.
Due to symmetry of each worker group, the cost can be calculated by aggregating the within-group
communication cost incurred by each basic partition plan:
\begin{equation}
\text{cost}(P)=\sum_{i=1}^k 2^{i-1}\text{cost}(p_i)
\end{equation}

We can then show that the following commutativity property holds:
\begin{theorem}\label{thm:comm}
$\text{cost}(\langle p_1,p_2 \rangle)=\text{cost}(\langle p_2,p_1 \rangle)$,
where $p_1$ and $p_2$ are basic partition plans.
\end{theorem}
\begin{proof}
The case is trivial if $p_1=p_2$. Let $G$ be the unpartitioned
dataflow graph; $G_{1}$ and $G_{2}$ be the partitioned graph
by $p_1$ and $p_2$; Let $\vec{S_1}$ and $\vec{S_2}$ be the tensor
size vectors of  $G_{1}$ and $G_{2}$ defined in lemma~\ref{thm:prop}.
Because every tensor is partitioned
by half, $\vec{S_1}=\vec{S_2}=\frac{1}{2}\vec{S}$. By lemma~\ref{thm:prop},
we then have:
\[
\begin{array}{rl}
\text{cost}(\langle p_1,p_2 \rangle) &= \text{cost}(p_1) + 2 * \text{cost}(p_2) \\
&= \vec{\alpha_1}\cdot\vec{S} + 2 * \vec{\alpha_2}\cdot\vec{S_1} \\
&= \vec{\alpha_1}\cdot 2\vec{S_2} + 2 * \vec{\alpha_2}\cdot\frac{1}{2}\vec{S} \\
&= \vec{\alpha_2}\cdot\vec{S} + 2 * \vec{\alpha_1}\cdot\vec{S_2}\\
&= \text{cost}(p_2) + 2 * \text{cost}(p_1) \\
&= \text{cost}(\langle p_2,p_1 \rangle)
\end{array}
\]
\end{proof}

Let the per-step cost be $\delta_i=2^{i-1}\text{cost}(p_i)$. We can easily
prove theorem~\ref{thm:greedy}.
\begin{theorem}\label{thm:greedy}
	Let the total communication cost incurred by all worker groups at step $i$ be $\delta_i$. Then $\delta_i\leq \delta_{i+1}$.
\end{theorem}

\begin{proof}
We prove by contradiction.
Suppose there exists a sequence $\langle p_1,\ldots,p_i,p_{i+1}\rangle$ such that
$\delta_i >\delta_{i+1}$. By theorem~\ref{thm:comm},
\[
\text{cost}(\langle p_1,\ldots,p_i,p_{i+1}\rangle)=\text{cost}(\langle p_1,\ldots,p_{i+1},p_i\rangle)
\]
Because $\delta_i >\delta_{i+1}$, we have
\[
\text{cost}(\langle p_1,\ldots,p_{i+1}\rangle)<\text{cost}(\langle p_1,\ldots,p_i\rangle)
\]
This means applying $p_{i+1}$ instead of $p_i$ at step $i$ is a more optimized partitioning,
which contradicts with the per-step optimality of the dynamic programming algorithm.
\end{proof}

\subsection{Optimiality proof}
\begin{theorem}
The recursive algorithm is optimal.
\end{theorem}
\begin{proof}
Let $P=\langle p_1, p_2,\ldots, p_k \rangle$ be the partition sequence
produced by the recursive algorithm and $O=\langle o_1, o_2,\dots,o_k \rangle$
be the optimal sequence. By theorem~\ref{thm:comm},
we can reorder the sequence so that the per-step costs
of both sequences are non-descending.

We prove by contradiction. Suppose $\text{cost}(P)>\text{cost}(O)$.
Then there must exist a step $i$ such that:
\begin{equation}
\text{cost}(\langle p_1, \ldots, p_i \rangle) \leqslant \text{cost}(\langle o_1, \ldots, o_i \rangle)
\end{equation}
\begin{equation}
\text{cost}(\langle p_1, \ldots, p_i, p_{i+1} \rangle) > \text{cost}(\langle o_1, \ldots, o_i, o_{i+1} \rangle)
\end{equation}

Let $G_p$ and $G_o$ be the partitioned dataflow graphs after applying
$\langle p_1, \ldots, p_i \rangle$
and $\langle o_1, \ldots, o_i \rangle$, respectively. Every tensor
is only $2^i$ of the size of original tensor so
$\vec{S_p}=\vec{S_o}$. Finally, by lemma~\ref{thm:prop}, we have
\begin{equation*}
\begin{array}{rl}
\text{cost}(\langle a_1, \ldots, a_i, a_{i+1} \rangle) &> \text{cost}(\langle o_1, \ldots, o_i, o_{i+1} \rangle) \\
&= \text{cost}(\langle o_1, \ldots, o_i\rangle)+2^i \vec{\alpha_o}\cdot\vec{S_o} \\
&\geq \text{cost}(\langle a_1, \ldots, a_i \rangle) + 2^i \vec{\alpha_o}\cdot\vec{S_a} \\
&= \text{cost}(\langle a_1, \ldots, a_i, o_{i+1} \rangle)
\end{array}
\end{equation*}

Hence, applying $o_{i+1}$ at step $i+1$ produces strictly less communication
cost than applying $a_{i+1}$, which contradicts the per-step optimality of the dynamic programming algorithm.

\end{proof}

\end{document}